%% file: mainTCS.tex
\documentclass[a4paper]{article}

\usepackage{enumitem}
\usepackage{complexity}

\usepackage[utf8]{inputenc}
\usepackage[T1]{fontenc}
\usepackage{amsmath,amsthm,amssymb,amsfonts,graph}
\usepackage{graphicx}
\usepackage{xspace}
\usepackage{hyperref}
\usepackage{cleveref}
\usepackage{verbatim}
\usepackage{gensymb}
\usepackage{tikz}
\usepackage{setspace}
\usetikzlibrary{snakes}
\usepackage[procnumbered,linesnumbered,ruled,vlined]{algorithm2e}
\usepackage[framemethod=TikZ]{mdframed}

\usepackage{cellspace}
\newcommand{\newvert}{\mathsf{new}}
\newcommand{\Eleft}[1]{#1_{\mathsf{left}}}
\newcommand{\Eleftedge}{\Eleft{E}}
\newcommand{\Eright}[1]{#1_{\mathsf{right}}}
\newcommand{\Erightedge}{\Eright{E}}
\newcommand{\Ecross}[1]{#1_{\mathsf{cross}}}
\newcommand{\Ecrossedge}{\Ecross{E}}
\newcommand{\Ileft}{\mathfrak{I}_{\mathsf{left}}}

\newcommand{\apexgraph}[1]{#1_{k\operatorname{-{\mathsf{div}}}}^{\mathsf{\;*}}}

\newcommand{\apexgraphthree}[1]{#1_{3\operatorname{-{\mathsf{div}}}}^{\mathsf{\;*}}}

\mdfdefinestyle{MyFrame}{
	linecolor=black,
	innerbottommargin=\baselineskip,
	backgroundcolor=gray!5!white}

\newcommand{\StabGIG}{\textup{\textsc{StabGIG}}}

\usepackage[textsize=footnotesize,color=green!40]{todonotes}
\usepackage[hmargin=2.5cm,vmargin=3cm]{geometry}
\usepackage{marvosym}

\newcommand{\curve}[1]{\mathtt{c}(#1)}
\newcommand{\lefty}{\mathbf{j}}
\newcommand{\righty}{\mathbf{i}}
\newcommand{\level}[1]{\mathsf{d}_{#1}}

\newcommand{\ie}{\textit{i.e.}}

\newcommand{\etal}{\textit{et al. \xspace}}

\newcommand{\subdivide}[1]{#1_{k\operatorname{-{\mathsf{div}}}}}

\newcommand{\LS}{\RoundedLsteel}
\newcommand{\U}{\sqcup}
\newcommand{\Urev}{\sqcap}
\newcommand{\HS}{\rotatebox[origin=c]{180}{\LS}}
\newcommand{\RS}{\reflectbox{\LS}}

\definecolor{dartmouthgreen}{rgb}{0.05, 0.5, 0.06}

\newtheorem{theorem}{Theorem}

\newtheorem{proposition}[theorem]{Proposition}

\newtheorem{observation}[theorem]{Observation}
\newtheorem{corollary}[theorem]{Corollary}

\newtheorem{question}{Question}

\newtheorem{remark}{Remark}

\newcommand{\Pb}[4]{%
\begin{center}
  \begin{tabular}{|l|}%
  \hline
    \begin{minipage}[c]{0.95\textwidth}
      \smallskip%
      \par\noindent%
      #1%
      \par\noindent%
      \textbf{\textsf{Input}}: #2%
      \par\noindent%
      \textbf{\textsf{#3}}: #4 
      \smallskip%
      \par\noindent%
    \end{minipage}
  \\\hline
  \end{tabular}%
\end{center}
}%

  \title{Recognizing Geometric Intersection Graphs Stabbed by a Line}

\author{Dibyayan Chakraborty\footnote{Univ Lyon, CNRS, ENS de Lyon, Université Claude Bernard Lyon 1, LIP UMR5668, France}
\and Kshitij Gajjar\footnote{Indian Institute of Technology Jodhpur, NH 62, Karwar, Jodhpur, Rajasthan, India -- 342037} \and Irena Rusu\footnote{Nantes Université, École Centrale Nantes,
CNRS, LS2N, UMR 6004, F-44000 Nantes, France}}

\begin{document}

\maketitle

\setstretch{1.1}
\begin{abstract}
In this paper, we determine the computational complexity of recognizing two graph classes, \emph{grounded \LS}-graphs and \emph{stabbable grid intersection} graphs. An \LS-shape is made by joining the bottom end-point of a vertical ($\vert$) segment to the left end-point of a horizontal ($-$) segment. The top end-point of the vertical segment is known as the {\em anchor} of the \LS-shape. Grounded \LS-graphs are the intersection graphs of \LS-shapes such that all the \LS-shapes' anchors lie on the same horizontal line. We show that recognizing grounded \LS-graphs 
is $\NP$-complete. This answers an open question asked by Jel{\'\i}nek \& T{\"o}pfer (Electron. J. Comb., 2019).



Grid intersection graphs are the intersection graphs of axis-parallel line segments in which two vertical (similarly, two horizontal) segments cannot intersect. We say that a (not necessarily axis-parallel) straight line $\ell$ stabs a segment $s$, if $s$ intersects $\ell$. A graph $G$ is a stabbable grid intersection graph ($\StabGIG$) if there is a grid intersection representation of $G$ in which the same line stabs all its segments. We show that recognizing $\StabGIG$ graphs is $\NP$-complete, even on a restricted class of graphs.  This answers an open question asked by Chaplick \etal (\textsc{O}rder, 2018).

\end{abstract}

\section{Introduction}

Recognizing a graph class means deciding whether a given graph is a member of the graph class. In this paper, we deal with the computational complexity of recognizing intersection graphs of certain types of geometric objects in the plane. These recognition problems stemmed from various fields of active research, like VLSI design~\cite{chung1983diogenes,chung1987embedding,sherwani2007}, map labelling~\cite{agarwal1998label}, wireless networks~\cite{kuhn2008ad}, computational biology~\cite{xu2006fast}, and have now become an indelible part of computational geometry. 

Perhaps the most extensively studied class of geometric intersection graphs is the class of \emph{interval graphs} (intersection graphs of intervals on the real line),  introduced by Benzer~\cite{benzer1959topology} while studying the fine structure of genes. In 1962, Lekkerkerker \& Boland~\cite{lekkerkerker} proved that the class of interval graphs is precisely the class of graphs without \emph{holes}\footnote{A hole is an induced cycle with $4$ or more vertices.} and \emph{asteroidal triples}\footnote{Three vertices of a graph form an asteroidal triple if the removal of any one of the vertices (along with all its neighbouring vertices) from the graph does not disconnect the other two.}. The above elegant result of Lekkerkerker \& Boland~\cite{lekkerkerker} motivated researchers to study further generalizations of interval graphs and their characterizations. One such generalization was considered 
by Gy{\'a}rf{\'a}s and Lehel~\cite{gyarfas1985covering} who considered the class of \emph{\LS-intersection graphs}. An \emph{\LS-shape} \cite{gyarfas1985covering} is a couple made of a vertical and horizontal segment, whose bottom and left end-points coincide. \emph{\LS-intersection graphs} are the intersection graphs of \LS-shapes. Observe that interval graphs can be expressed as intersection graphs of \LS-shapes as follows: given an interval representation of a graph, replace each interval with an \LS-shape whose horizontal segment is the same as the interval on the horizontal line $y=0$. The popularity of \LS-graphs among graph theorists increased when Gon{\c{c}}alves et al.~\cite{gonccalves2018planar} proved that all planar graphs (graphs that can be drawn in the plane in such a way that edges can intersect only at their end-points) are \LS-graphs. In contrast, recognizing \LS-graphs is $\NP$-complete~\cite{chakraborty2022}.


An interesting subclass of \LS-graphs called \emph{infinite-\LS-graphs} was considered by McGuinness~\cite{mcguinness1996bounding}. The top end-point of the vertical segment of an \LS-shape is called the {\em anchor} of the \LS-shape. \emph{Infinite-\LS-graphs}, also known as \emph{grounded \LS-graphs} are the intersection graphs of \LS-shapes whose anchors belong to the same horizontal line, called the \emph{ground line}. See Figure~\ref{fig:example-stab} for an example. Interval graphs are (also) grounded \LS-graphs. (Indeed, given a set of intervals, replace each interval with an \LS-shape whose horizontal segment is the same as the interval on the horizontal line $y=0$ and whose anchor lies on the horizontal line $y=1$). Other well-studied subclasses of grounded \LS-graphs are (see~\cite{asinowski2012vertex}) outerplanar graphs, permutation graphs, circle graphs etc. Researchers have studied different aspects (e.g. chromatic number~\cite{mcguinness1996bounding,davies2021grounded}, dominating set~\cite{chakraborty2022dominating}, independent set~\cite{bose2022computing}, forbidden patterns~\cite{jelinek2019grounded}) of grounded \LS-graphs. However, the complexity of recognizing grounded \LS-graphs remained open. In this paper, we prove that recognizing grounded \LS-graphs is $\NP$-complete answering a question asked by Jel{\'\i}nek and T{\"o}pfer~\cite{jelinek2019grounded}.

\begin{mdframed}[style=MyFrame]
\vspace{-0.2cm}
\begin{theorem}\label{thm:grounded-L}
Recognizing grounded \LS-graphs is $\NP$-complete.
\end{theorem}
\vspace{-0.2cm}
\end{mdframed}

\begin{figure}
    \centering
    \scalebox{0.9}{
    \begin{tabular}{cp{0.5cm}cp{0.5cm}c}
       \includegraphics[scale=0.6]{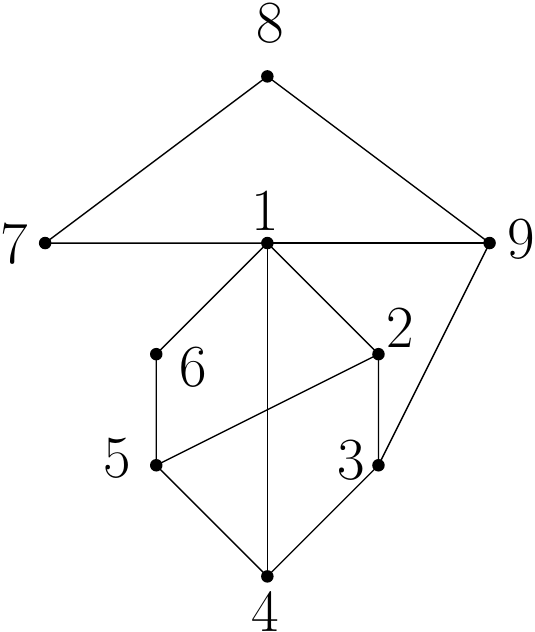}  && 
       \scalebox{0.35}{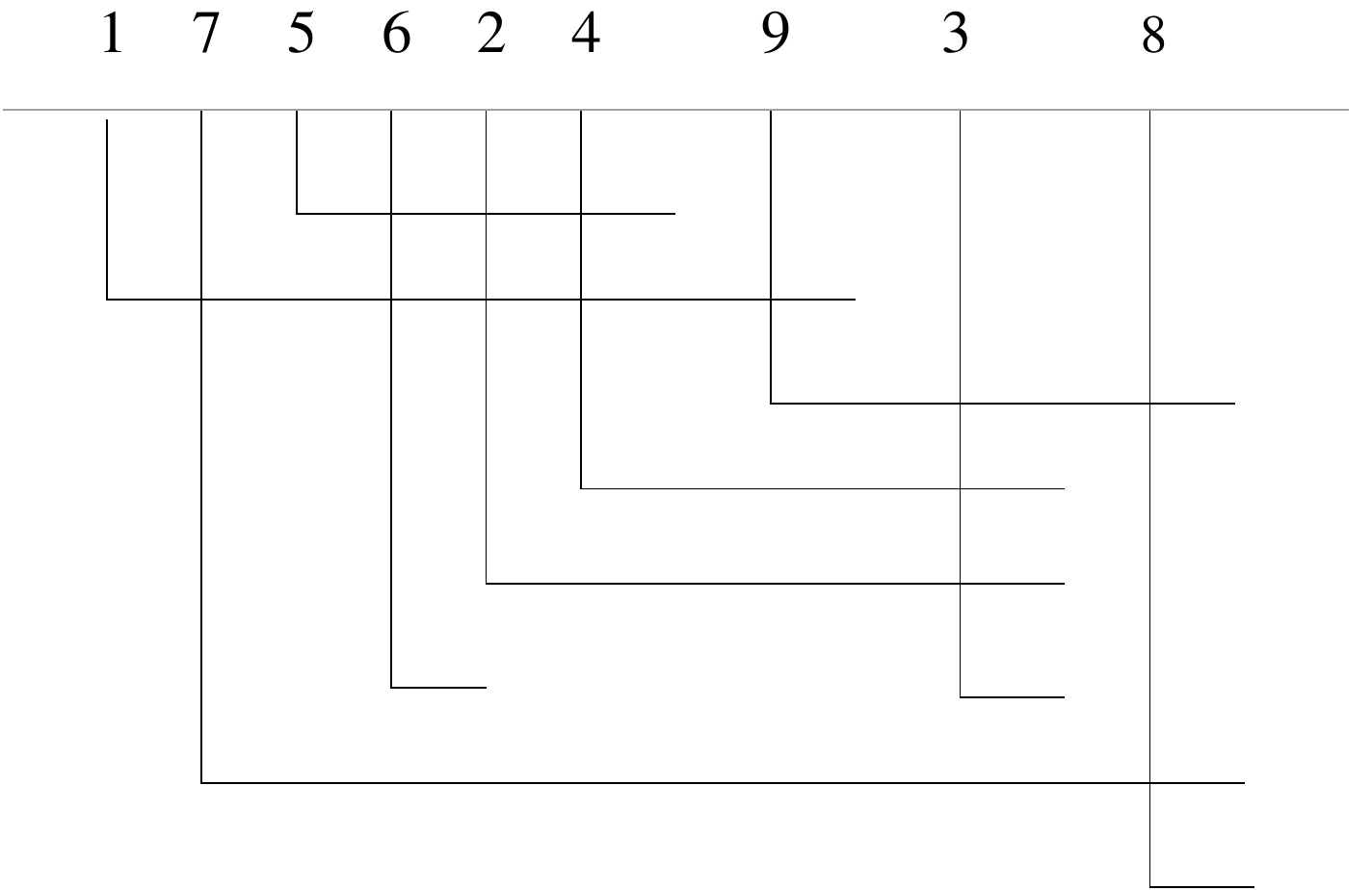} &&
       \includegraphics[scale=0.5]{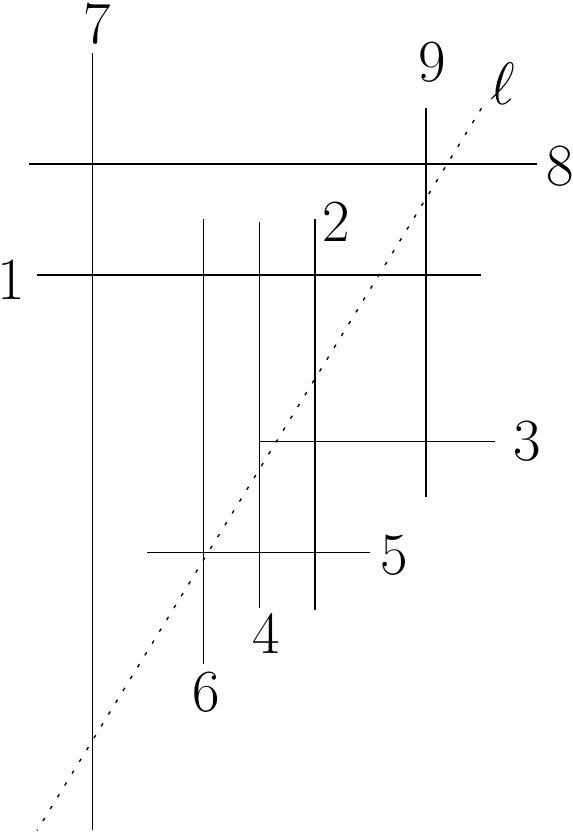} \\
         &  &
    \end{tabular}}
    \caption{A graph (left) with its grounded \LS-shape representation (middle) and  its stabbable grid intersection representation (right).} \label{fig:example-stab}
\end{figure}

\emph{Rectangle intersection graphs} are graphs with boxicity $2$ \ie~intersection graphs of axis-parallel rectangles in the plane. This graph class was introduced by Asplund \& Gr\"{u}nbaum~\cite{asplund1960coloring} in 1960 (even before the introduction of boxicity), who studied the chromatic number of such graphs. After almost three decades, Kratochv{\'\i}l~\cite{kratochvil1994} proved that recognizing rectangle intersection graphs is $\NP$-complete, even if the input graphs are restricted to bipartite graphs. Interestingly, bipartite rectangle intersection graphs are exactly the class of grid intersection graphs (\textsc{GIG}), \ie, intersection graphs of axis-parallel line segments in the plane where no two segments with the same orientation intersect~\cite{hartman1991}. This implies that the recognition of \textsc{GIG} is $\NP$-complete, and motivates the study of its subclasses. Chaplick \etal~\cite{chaplick2018grid} introduced the class of \emph{stabbable grid intersection graphs}. A segment $s$ in the plane is \emph{stabbed} by a line $\ell$ if $s$ and $\ell$  intersect. A graph $G$ is a \emph{stabbable grid intersection graph} ($\StabGIG$) if it has a grid intersection representation such that there exists a straight line that stabs all the segments of the representation. See Figure~\ref{fig:example-stab} for an example. All planar bipartite graphs are $\StabGIG$s~\cite{felsner2013rectangle} and Chaplick \etal\cite{chaplick2018grid} left the problem of recognizing stabbable grid intersection graphs as open. In this paper, we answer their question.


\begin{mdframed}[style=MyFrame]
\vspace{-0.2cm}
\begin{theorem}\label{thm:stabbable-GIG}
Recognizing $\StabGIG$ is $\NP$-complete.
\end{theorem}
\vspace{-0.2cm}
\end{mdframed}
We note that recently the computational complexities of several geometric intersection graph classes with "stabbed" representations have been settled. Examples include \emph{stick} graphs, \emph{bipartite hook} graphs and \emph{max point-tolerance} graphs~\cite{rusu2022complexity}. Grounded \LS-graphs and  stabbable grid intersection graphs were the largest classes for which the complexities were still unknown. The recognition problem is still open for some
smaller classes of graphs with "stabbed" representations (see Section~\ref{sec:conclude}).

\medskip\noindent\textbf{Organisation:} In Sections~\ref{sec:grounded-L} and~\ref{sec:stabgigrep}, we prove Theorem~\ref{thm:grounded-L} and Theorem~\ref{thm:stabbable-GIG}, respectively. In Section~\ref{sec:conclude}, we conclude.

\medskip\noindent\textbf{Notations:} For a positive integer $n\geq 1$, $[n]$ denotes the set $\{1,2, \ldots, n\}$. All graphs considered in this paper are simple and undirected. For a graph $G$, the sets $V(G)$ and $E(G)$ denote the vertex set and edge set of $G$, respectively. 

\medskip\noindent\textbf{Convention:} Both for grounded \LS-graphs and StabGIGs, we assume without loss of generality that any two segments in the representation intersect in at most one point. Otherwise, slight changes in the position and length of the segments allow to transform a representation that does not have this property into a representation that has this property.

\section{Proof of Theorem~\ref{thm:grounded-L}}\label{sec:grounded-L}

In our presentation, the geometrical representation of the grounded \LS-graphs uses \HS-shapes above the horizontal ground line rather than \LS-shapes below it. For a grounded \LS-graph, this representation is its {\em grounded \HS-representation}. Note that the convention above implies that all the anchors are distinct.

The \HS-shape representing a vertex $x$ is denoted by $\HS(x)$, and the anchor of $\HS(x)$, now located at the bottom of the vertical segment, 
is denoted by $x$ too. The left to right order of the anchors along the ground line is denoted
by $\prec$. Note that for two intersecting \HS-shapes \HS$(x)$ and \HS$(y)$, we have  $x\prec y$ if and only if the vertical segment of
\HS$(x)$ intersects the horizontal segment of \HS$(y)$. A \HS-shape that has intersections only on its vertical (respectively horizontal) segment
is called a {\em \HS$^v$-shape} (respectively a {\em \HS$^h$-shape}). 
Then the problem we are interested in may be restated as follows:

\Pb
{{\sc Grounded \LS-Graphs Recognition (Grounded \LS-Rec)}}
{A graph $H$.}
{Question}
{Is there a  grounded \HS-representation for $H$?}

Theorem \ref{thm:grounded-L} is equivalent to the statement that {\sc Grounded \LS-Rec} is $\NP$-complete. In order to show it, we use another class of geometric intersection graphs, namely stick graphs. They have been defined by Chaplick et al.~\cite{chaplick2018grid}
as the intersection graphs of a set $A$ of vertical segments in the plane and a set $B$ of horizontal segments in the plane, 
such that the bottom end-point of the segments in $A$ and the left end-point of the segments in $B$ 
belong to a ground straight line with slope -1. Again, all the endpoints may be considered as distinct.
Stick graphs are bipartite graphs, and the aforementioned geometrical representation is called their {\em stick representation}. We prove Theorem \ref{thm:grounded-L} using a reduction from the problem below, which is $\NP$-complete~\cite{rusu2022complexity}:

\Pb
{{\sc Stick Graphs Recognition (StickRec)}}
{A bipartite graph $G=(A\cup B, E)$.}
{Question}
{Is there a  Stick representation of $G$ with sets $A$ and $B$?}

Grounded \LS-graphs and stick graphs are related by the relationships given in
Propositions \ref{prop:StickToL} and \ref{prop:LtoStick} below, which are essential to our construction. A grounded \HS-representation of a bipartite graph $G=(A\cup B, E)$ is called {\em nice} if  the vertices in $A$ are represented  by \HS$^h$-shapes, and those in $B$ by \HS$^v$-shapes.
We use the notations ${\mathsf H}(a)$  and ${\mathsf V}(c)$ for the horizontal and vertical segments of \HS$(c)$ respectively, where $c\in A\cup B$.

 \begin{figure}[t!]
\centering
\scalebox{0.7}{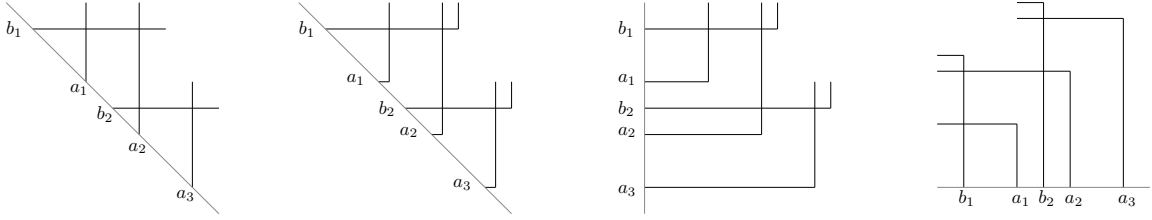}
\caption{\small Transforming a stick representation into a grounded \HS-representation. Here, the stick graph is the induced path $a_1b_1a_2b_2a_3$.}
\label{fig:StickToL}
\end{figure}

\begin{proposition}
 All stick graphs  are grounded \LS-graphs. Moreover, each stick graph  has a nice grounded \HS-representation.
 \label{prop:StickToL}
\end{proposition}

\begin{proof} Let $G$ be a stick graph. Figure \ref{fig:StickToL} illustrates the steps of the proof. Consider a stick representation of $G$, and replace first each vertical (respectively horizontal) segment 
 with a \RS-shape having the same anchor as the initial segment, whose horizontal (respectively vertical) segment is not used for intersections.  Then extend the horizontal segments of the  \RS-shapes towards left so that to place all the anchors on the same vertical line, that becomes the new ground line. Finally, perform a 90 degrees counterclockwise rotation
 of the whole representation to obtain a grounded \HS-representation of $G$. In this representation, the vertices in $A$ are represented  
 by \HS$^h$-shapes, and those in $B$ by \HS$^v$-shapes.
\end{proof}

\begin{remark}
Not all the bipartite grounded \LS-graphs are stick graphs. The bipartite graph in Figure \ref{fig:NotStick} is not a stick graph  \cite{chaplick2018grid}, but is a grounded \LS-graph as shown by the \HS-representation we provide.
\end{remark}

\begin{proposition}
All bipartite grounded \LS-graphs which have a nice grounded \HS-representation  are stick graphs.
 \label{prop:LtoStick}
\end{proposition}

\begin{proof}
Let $G$ be a bipartite graph with  a nice grounded \HS-representation $R$.  We transform this representation into a stick representation of $G$ as follows.
Draw a line $\ell$ with slope 1 that intersects the grounded line to the left of all the anchors, and such that all the
segments ${\mathsf H}(a)$, for $a\in A$, lie in the lower half-plane defined by $\ell$.  Such a line always exists, it is sufficient to consider the intersection point sufficiently far to the left. Consider a Cartesian coordinate system whose origin is the intersection point of the two lines, the $x$-axis is the grounded line and the $y$-axis is oriented upwards.   

For each $c\in A\cup B$, denote by $l(c)$ and $r(c)$  the $x$-coordinate of the left and right endpoint of ${\mathsf  H}(c)$ respectively, and by $h(c)$ the common $y$-coordinate of these two endpoints. For $c,c'\in A\cup B$, we say that ${\mathsf  H}(c)$ {\em is below} ${\mathsf  H}(c')$ (equivalently, ${\mathsf  H}(c')$ is {\em above} ${\mathsf  H}(c)$) whenever $h(c)<h(c')$ and the projections of ${\mathsf  H}(c)$ and ${\mathsf  H}(c')$ on the $x$-axis intersect.

Since $\ell$ is above ${\mathsf  H}(a)$, for $a\in A$, we deduce that $h(a)<l(a)<r(a)$.  We show that the heights of the \HS-shapes can be modified so as to place the right endpoint of ${\mathsf  H}(a)$ on the line $\ell$, for each $a\in A$. Equivalently, we show that: 

\bigskip

($P$) There exists a grounded \HS-representation $R'$ of $G$ such that $h(a)=r(a)$ for each $a\in A$.
\bigskip

Let $a_1, a_2, \ldots, a_{|A|}$ be the vertices in $A$, in decreasing order of their value $r(a)$ in the grounded \HS-representation $R$. We use induction to show that, for each $1 \leq i\leq |A|$, there exists a grounded \HS-representation $R_i$ of $G$ such that $h(a_k)=r(a_k)$ for $1 \leq k\leq i$. Note that the for each $u,v$ such that $l(a_u)<r(a_v)<r(a_u)$ we necessarily have $h(a_u)>h(a_v)$. In the contrary case, \HS$(a_u)$ and \HS$(a_v)$ would intersect, which is impossible since $G$ is bipartite and $A$ is one of its parts. Then, the order $a_1, a_2, \ldots, a_{|A|}$ of the
vertices implies that, for $u<v$, either the projections ${\mathsf  H}(a_v)$ and ${\mathsf  H}(a_u)$ on the $x$-axis are disjoint and in this order from left to right, or ${\mathsf  H}(a_v)$ is below ${\mathsf  H}(a_u)$.

Let $i=1$. As noticed above, no ${\mathsf  H}(a_j)$ with $j\neq 1$ is above ${\mathsf  H}(a_1)$. We then modify \HS$(a_1)$ so that $h(a_1)$ becomes equal to $r(a_1)$, and appropriately lengthen the vertical segment of \HS$(b)$ for each neighbor $b$ of $a_1$ so that \HS$(b)$ intersects \HS$(a_1)$ (equivalently, such that ${\mathsf H}(b)$ is above ${\mathsf H}(a_1)$). Note that this transformation does not modify the intersections between \HS$(b)$ and the other \HS-shapes, since the segment ${\mathsf  V}(b)$ -- which realizes
by hypothesis all the intersections of \HS$(b)$ with other \HS-shapes -- is neither moved (to left or right) nor shortened. 
Therefore, this is a new grounded \HS-representation of $G$, denoted $R_1$. 

Assume now, by inductive hypothesis, that the grounded \HS-representation $R_i$ of $G$ exists. In $R_i$, the only segments ${\mathsf  H}(a_j)$ that are possibly placed above ${\mathsf  H}(a_{i+1})$ satisfy $j\leq i$, and by the inductive hypothesis they also satisfy $h(a_j)=r(a_j)$. 
The ordering of the vertices in $A$ implies that $r(a_j)>r(a_{i+1})$. As above,  we  modify \HS$(a_{i+1})$ so that $h(a_{i+1})$ becomes equal to $r(a_{i+1})$. Then $h(a_{i+1})=r(a_{i+1})<r(a_j)=h(a_j)$, meaning that every ${\mathsf  H}(a_j)$ which was previously above ${\mathsf  H}(a_{i+1})$ is still above
${\mathsf  H}(a_{i+1})$, so that no wrong intersection is created between \HS$(a_{i+1})$ and \HS$(a_j)$. Furthermore, for the neighbors $b$ of $a_{i+1}$
whose horizontal segment is now below ${\mathsf  H}(a_{i+1})$, the segment ${\mathsf  V}(b)$ is appropriately lengthened such that $h(a_{i+1})<h(b)<h(a_i)$
(equivalently, such that ${\mathsf H}(b)$ is above ${\mathsf H}(a_{i+1})$
and below ${\mathsf H}(a_{i})$).
Similarly to the case $i=1$, these modifications do not modify the existing intersections between \HS-shapes, so that we obtain the sought
grounded \HS-representation of $G$ denoted $R_{i+1}$. 

Then ($P$) is proved, with $R'=R_{|A|}$.
To obtain a stick representation of $G$, in $R'$ consider only the line $\ell$ as well as the segments ${\mathsf  H}(a)$ for $a\in A$ and 
the sub-segment of ${\mathsf  V}(b)$ that lies above $\ell$, for $b\in B$.
A 90 degrees clockwise rotation of this representation is a stick representation of $G$ (with a ground line of slope -1).
\end{proof}

 \begin{figure}[t!]
\centering
\scalebox{0.4}{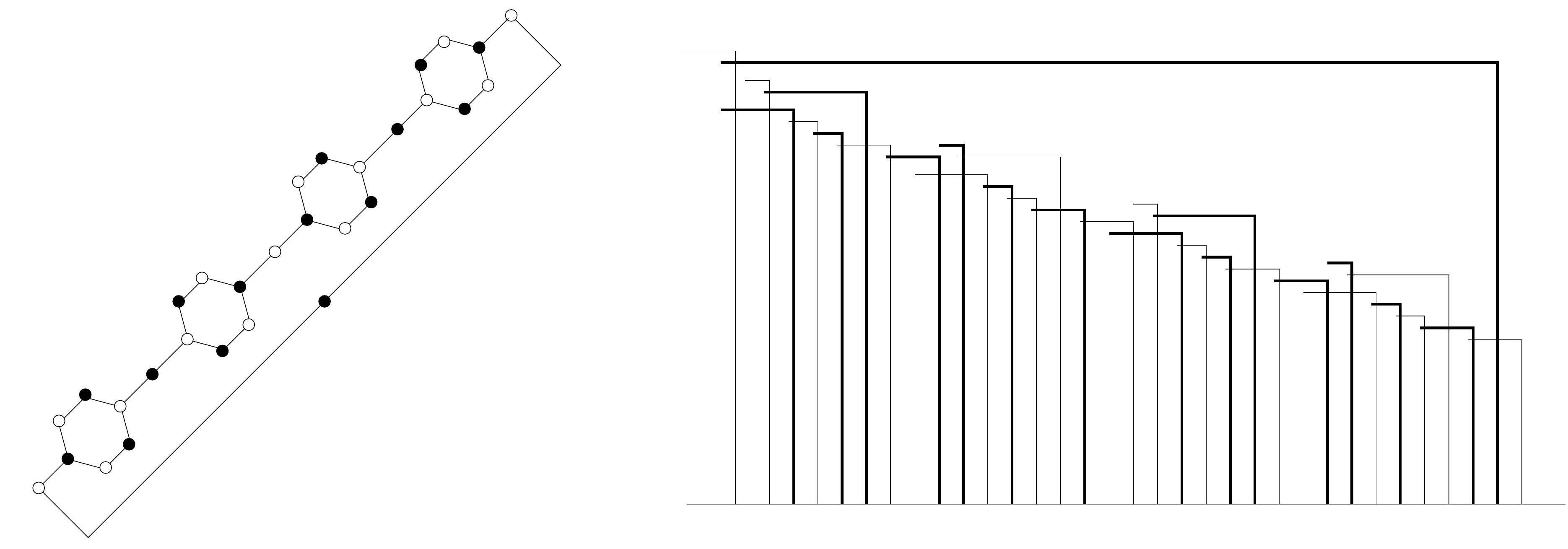}
\caption{\small An example of a bipartite grounded \LS-graph, which is not a stick graph. The bipartition is indicated by the two colours of the vertices, white and black. In the grounded \HS-representation presented here, the thick \HS-shapes correspond to the black vertices.}
\label{fig:NotStick}
\end{figure}

\subsection{The reduction}

Our reduction from {\sc StickRec} to {\sc Grounded \LS-Rec} transforms a given bipartite graph $G=(A\cup B,E)$ into a (non-bipartite) graph $H$ that is roughly a copy of $G$ where each vertex $x$ from $B$ is placed inside a gadget graph called $\Lambda(x)$. The role of the gadgets is to force the grounded \HS-representation of $G$ produced inside any grounded \HS-representation of $H$ (if any) to be nice, that is, to represent the vertices in $A$ by \HS$^h$-shapes and the vertices in $B$ by \HS$^v$-shapes. In this way, using Propositions~\ref{prop:StickToL} and \ref{prop:LtoStick}, a grounded \HS-representation exists for $H$ if and only if a stick representation exists for $G$.

More formally, let $\Lambda(x)$ be the graph in Figure \ref{fig:GammaAndReduc} (left). The vertex denoted $x$ also has a second label, namely $2$, which is easier to use in the proof of Proposition \ref{prop:GammaUnique}, rather than $x$. The reduction is simple, and is defined below. Consider an instance $G=(A\cup B, E)$ of {\sc StickRec}, and build an instance $H$ of {\sc Grounded \LS-Rec} as follows (see Figure \ref{fig:GammaAndReduc} (right) for an example):

\begin{enumerate}
 \item for each $a\in A$, define a vertex $a$ of $H$. 
 \item for each $b\in B$, include $\Lambda(b)$ into $H$. For each pair $b\neq b'$, the graphs $\Lambda(b)$ and  $\Lambda(b')$ are
 vertex-disjoint.
 \item for each edge $ab\in E$ with $a\in A$ and $b\in B$, add an edge from $a$ to every vertex in $\Lambda(b)$.
\end{enumerate}

A vertex $u$ that is adjacent to every vertex in $\Lambda(x)$ is called {\em universal} with respect to $\Lambda(x)$. 
\newpage

We show in the next section that this reduction proves the $\NP$-completeness of {\sc Grounded \LS-Rec}. To this end, we follow the following steps:\nopagebreak
\begin{enumerate}
    \item[$(i)$] \label{enum:LambdaUnique} We show in Proposition \ref{prop:GammaUnique} that $\Lambda(x)$ accepts exactly two grounded \HS-representations, shown in Figure \ref{fig:Gamma2Rep}(c) (details are given later) and roughly represented as in Figure \ref{fig:Gamma2Rep}(d).
     \item[$(ii)$] \label{enum:universal} We show in Proposition \ref{prop:universal} that adding to $\Lambda(x)$ a universal vertex $u$ yields another grounded \LS-graph, whose grounded \HS-representations always place the anchor of \HS$(u)$ outside the representation of $\Lambda(x)$, to the right of it.  Consequently, $x$ ($u$ respectively) is represented by a \HS$^v$-shape (\HS$^h$-shape respectively). Recalling that in our construction of $H$ each vertex $a\in A$ is universal with respect to $\Lambda(b)$, for each of its neighbors $b$ from $G$, this implies that the grounded
     \HS-representations of $H$ mimic those of $G$ and vice-versa, as shown in Figure \ref{fig:GroundedH}.
     \item[$(iii)$] \label{enum:iff} We show in Proposition \ref{prop:yes} that $G$ is a stick graph if and only if $H$ is a grounded \LS-graph.
     \item[$(iv)$] \label{enum:NPc} We conclude the proof of Theorem \ref{thm:grounded-L}.
\end{enumerate}

 \begin{figure}[t!]
\centering
\scalebox{0.45}{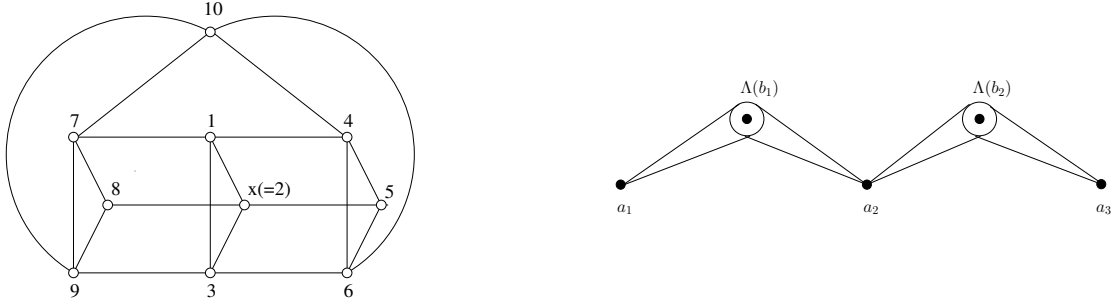}
\caption{\small The graph $\Lambda(x)$ (left) and the construction of the graph $H$ when $G$ is the path $a_1b_1a_2b_2a_3$ (right).}
\label{fig:GammaAndReduc}
\end{figure}

\subsection{The proofs}\label{subsec:theproofs}

The graph $\Lambda(x)$ has been chosen so as to have very few and convenient grounded \HS-representations. In any grounded \HS-representation of $\Lambda(x)$, the left to right order of the anchors $1, 2, 3$ defines a permutation $def$ on $\{1,2,3\}$. For $i\in\{1,2,3\}$, let
 $v^1_i$ be the neighbor of $i$ from $\{7,8,9\}$ and $v^2_i$ the neighbor of $i$ from $\{4,5,6\}$.
 



\begin{proposition}
The graph $\Lambda(x)$ is a grounded \LS-graph. It admits exactly two grounded \HS-representations, given in Figure \ref{fig:Gamma2Rep}(c),
where we have either $d=1$ and $f=3$, or $d=3$ and $f=1$.
\label{prop:GammaUnique}
\end{proposition}

\begin{proof}

 When \HS$(d)$, \HS$(e)$ and \HS$(f)$ (only) are placed on the ground line such that $d\prec e\prec f$, the \HS-shapes of the two neighbors $v_e^1$ and $v_e^2$ of $e$ may be placed (independently of each other) in exactly two locations:
 
 \begin{enumerate}[itemsep =0cm]
 \item[$\bullet$] either to the left of the anchor $d$, 
 \item[$\bullet$] or to the right of the anchor $f$. 
 \end{enumerate}
 
 Indeed, we cannot have $e\prec v_e^i\prec f$, for some $i\in\{1,2\}$, as explained hereafter. If, for instance, $e\prec v_e^1\prec f$, then we must have $d\prec v_d^1, v_f^1, 10\prec e$ in order to satisfy all the adjacencies and non-adjacencies. Furthermore,  $d\prec v_f^2\prec e$ is not possible, since then \HS$(v_f^2)$ intersects either \HS$(v_d^1)$ or \HS$(v_e^1)$ before it intersects \HS$(f)$, a contradiction. We then have $v_f^2\prec d$ or $e\prec v_f^2$, which implies that $v_d^2$ and $v_e^2$ are
 also outside the interval defined by $d$ and $e$. But then 10 cannot have neighbors among $v_d^2, v_e^2$ and $v_f^2$, a contradiction.
 
 
 Two cases appear, for  which the proofs are similar. In each case, we first deduce the possible grounded \HS-representations of the subgraph $\Lambda^9$ induced  by the vertices 1 to 9, using the variables $i, v_i^1, v_i^2$ for $1\leq i\leq 3$. For each of these representations, we then test the possibility
 to add \HS$(10)$.

 \medskip



 {\bf Case 1.} $v_e^1$ and $v_e^2$ belong to different locations (see Figure~\ref{fig:Gamma2Rep}(a)). Assume w.l.o.g. that $v_e^2\prec d$ and $f\prec v_e^1$. Then, there is exactly one possible position for $v_f^2$ and two possible positions for $v_f^1$ (the anchors are denoted by $v_f^1$ and $v_f^{'1}$) so as
 to obtain the sought intersections with $f$ and $v_e^2, v_e^1$ respectively. Similarly, only one position is possible for $v_d^1$ and two positions for $v_d^2$ (the anchors are denoted by $v_d^2$ and $v_d^{'2}$). 
 Figure~\ref{fig:Gamma2Rep}(a)
 thus records all the possible representations of the graph induced by the vertices 1 to 9: any possible representation is obtained by choosing 
 one of $v_f^1, v_f^{'1}$ and one of $v_d^2, v_d^{'2}$. 
 
 The vertex $10$ is adjacent to $v_i^1, v_i^2, v_j^1, v_j^2$, with $i\neq j$ and $i,j\in\{d,e,f\}$.  As $d$ separates $v_e^1$ and $v_e^2$, and
 $10$ is not adjacent to $d$, we deduce that $i,j\neq e$. Moreover, $\{i,j\}=\{d,f\}$ is possible, as shown in Figure~\ref{fig:Gamma2Rep}(c) 
 where an appropriate position is proposed for \HS$(10)$. This position is unique and implies that the anchor $v_d^2$ (and not $v_d^{'2}$) has to
 be used, as well as the anchor $v_f^1$ (and not $v_f^{'1}$).
 
 Since $e$ is the unique vertex among $\{1,2,3\}$ none of
 whose neighbors is adjacent to $10$, we deduce that $e=2$, and thus $\{d,f\}=\{1,3\}$. In conclusion, we have two possible \HS-representations for 
 $\Lambda(x)$, with $d=1$ and $f=3$, and vice-versa.

     \begin{figure}[t]
\centering
\scalebox{0.3}{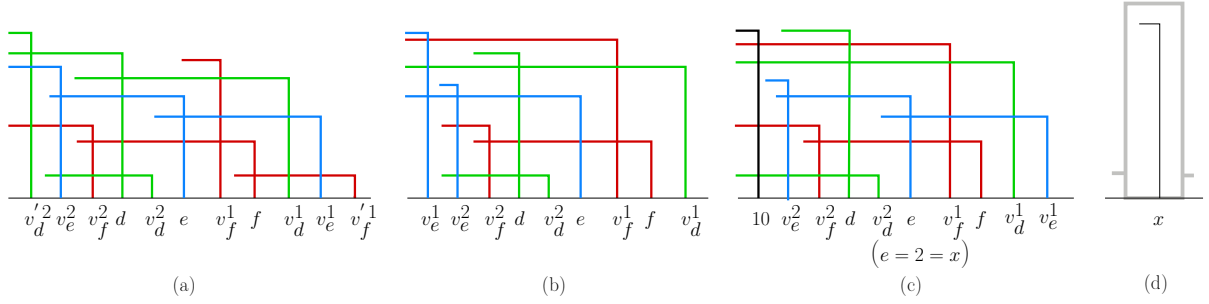}
\caption{\small Representations of $\Lambda(x)$: (a) Case 1 for $\Lambda^9$. (b) Case 2 for $\Lambda^9$. (c) Possible grounded \HS-representations for  $\Lambda(x)$. (c) The drawing replacing any representation, when only \HS$(x)$ and the level of the lowest horizontal segment are needed.}
\label{fig:Gamma2Rep}
\end{figure}

\medskip


{\bf Case 2.} $v_e^1$ and $v_e^2$ belong to the same location. If they are both on the right side of $f$, with (for instance) $f\prec v_e^1\prec v_e^2$, then there is no possible place for \HS$(v_d^1)$ so as to intersect \HS($v_e^1)$ and \HS$(d)$ but not $\HS(v_e^2)$ nor $e$. Thus 
 we may assume w.l.o.g that $v_e^1\prec v_e^2\prec d$. We cannot have $v_d^1\prec v_e^1$, since then \HS$(v_f^1)$ is impossible to place so as
 to intersect \HS$(f)$, \HS$(v_d^1)$ and $\HS(v_e^1)$, but no other \HS-shape. We deduce that $f\prec v_d^1$, and further that $e\prec v_f^1\prec f$, $v_e^2\prec v_f^2\prec d$ and $d\prec v_d^2\prec e$. The resulting order of the anchors for the vertices of $\Lambda^9$ is given in Figure~\ref{fig:Gamma2Rep}(b). 
 
 Looking for the neighbors  $v_i^1, v_i^2, v_j^1, v_j^2$ with $i\neq j$ and $i,j\in\{d,e,f\}$ of the vertex $10$, we notice that $e$ separates
 $v_d^1$ and $v_d^2$, as well as $v_f^1$ and $v_f^2$. None of these pairs of vertices can be adjacent to $10$ (since 10 is not adjacent to $e$), and thus there is no possible way
 to place \HS$(10)$ so as to intersect the vertices it needs to intersect. No \HS-representation for $\Lambda(x)$ can be found in this case.
\medskip

Consequently, all the grounded \HS-representations of $\Lambda(x)$ are obtained in Case 1, and are drawn in Figure~\ref{fig:Gamma2Rep}(c). 
\end{proof}

{\bf Convention.} In the remainder of the paper, the grounded \HS-representation in Figure \ref{fig:Gamma2Rep}(a) is
drawn as a grey box around \HS$(x)$, together with two lateral segments indicating the level of the lowest horizontal segment in the
representation. See Figure \ref{fig:Gamma2Rep}(d). 
\bigskip

Let $\Lambda(x)+u$ be the graph obtained by adding to $\Lambda(x)$ an universal vertex $u$.  

\begin{proposition}
The graph $\Lambda(x)+u$ is a grounded \LS-graph. In its grounded \HS-representations, $u$ is always the rightmost anchor,
and thus \HS$(u)$ is a \HS$^h$-shape.
\label{prop:universal}
\end{proposition}

\begin{proof}
A grounded \HS-representation of $\Lambda(x)+u$ is obtained by placing the anchor of $u$ to the extreme right of a grounded \HS-representation
of $\Lambda(x)$,  and letting the horizontal
segment of \HS$(u)$ intersect the vertical segments of all the other \HS-shapes at a very low height (lower than the horizontal
segment of \HS$(v_d^2)$). Then \HS$(u)$ is a \HS$^h$-shape. The conclusion of Proposition \ref{prop:universal} is reached by noticing 
that no other position is possible for $u$. Indeed, if $u$ was to the left of $v_e^1$, then \HS$(u)$ could intersect both \HS$(v_d^2)$ and
\HS$(v_e^1)$ only if $d\prec u\prec v_d^2$. But then the horizontal segment $\mathsf{H}(u)$ of \HS$(u)$ should be both below
$\mathsf{H}(v_f^2)$ (in order to intersect \HS$(v_f^2)$) and above $\mathsf{H}(v_f^1)$ (in order to intersect \HS$(v_f^1)$) and this is not possible.
\end{proof}

\begin{proposition}
The graph $G$ is a yes-instance of {\sc StickRec} if and only if the graph $H$ is a yes-instance of {\sc Grounded \LS-Rec}.
\label{prop:yes}
\end{proposition}

\begin{proof}
 For the forward direction, Proposition \ref{prop:StickToL} implies that the stick graph $G$ has a grounded \HS-represen\-tation. 
We transform this representation of $G$ into a grounded \HS-representation of $H$ as follows. For each $b\in B$, build a grounded \HS-representation of $\Lambda(b)$ (whose existence is ensured by Proposition~\ref{prop:GammaUnique}) as a block, in the immediate neighborhood of \HS$(b)$. See Figure \ref{fig:GroundedH}. The lowest horizontal segment of a \HS-shape in this grounded \HS-representation of $\Lambda(b)$,
indicated in the figure by the two grey lateral segments of the grey box, must be placed above all the horizontal segments of \HS-shapes \HS$(a)$ with $a\in A$ and $ab\in E$. This is always possible, by increasing the lengths of the vertical segments as needed. Furthermore, for each $a\in A$ with $ab\in E$, extend the horizontal segment of \HS$(a)$ such that it intersects all the \HS-shapes representing vertices in $\Lambda(b)$. The resulting grounded \HS-representation is a grounded \HS-representation of $H$.

We now consider the backward direction. Let $a\in A, b\in B$ such that $ab\in E$. By the construction of $H$, $a$ is universal for 
$\Lambda(b)$. By Proposition \ref{prop:universal} for the subgraph $\Lambda(b)+a$ of $H$, we deduce that the anchor of \HS$(a)$ is placed outside the representation of $\Lambda(b)$, to the right, so that \HS$(a)$ is a \HS$^h$-shape. It follows that the intersection between \HS$(b)$ and \HS$(a)$ holds on ${\mathsf V}(b)$ and on ${\mathsf H}(a)$, for each pair $a\in A, b\in B$ such that $ab\in E$. 

Then, when we focus only on the vertices in $G$,  the resulting grounded \HS-representation of $G$ is nice. We deduce by Proposition \ref{prop:LtoStick} that $G$ is a stick graph.
\end{proof}

\begin{figure}[t]
\centering
\scalebox{0.77}{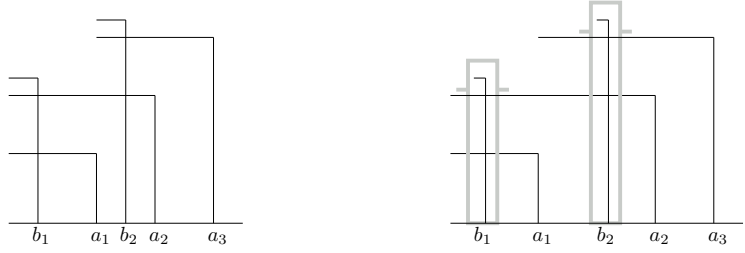}
\caption{\small The grounded \HS-representation of $G$ (left) and the grounded \HS-representation of $H$ (right), when $G$ is the induced path $a_1b_1a_2b_2a_3$.}
\label{fig:GroundedH}
\end{figure}

We are ready to prove that {\sc Grounded \LS-Rec} is $\NP$-complete.

\begin{proof}
(of Theorem \ref{thm:grounded-L}).
To show that {\sc Grounded \LS-Rec} is in $\NP$, we first note that a graph has a grounded \HS-representation if and only if it has a grounded \HS-representation whose segments have integer coordinates. For the forward direction, consider a graph that has a grounded \HS-representation with real coordinates for the 
segments. Order the $2n$ endpoints of the horizontal segments in increasing order of their
$x$-coordinate (equal values may appear, that yield the same rank). Remember the rank of each point (called $x$-rank). Do the same for the $y$-coordinate, and the $y$-rank. Then  replace the initial coordinates of each endpoint with its $x$-rank and its $y$-rank, and draw the \HS-shapes accordingly. The backward direction is immediate.  

In order to show that {\sc Grounded \LS-Rec} is in $\NP$, it is 
sufficient to test if a set of $n$ given \HS-shapes with integer coordinates defines a grounded \HS-representation for the input graph, and
this is done in polynomial time.
Proposition \ref{prop:yes} shows that our reduction is correct. Moreover, it takes
linear time to build $H$ knowing $G$, so that the reduction is polynomial.
\end{proof}

\section{Proof of Theorem~\ref{thm:stabbable-GIG}}\label{sec:stabgigrep}

In this section, we focus on the following problem:

\Pb
{{\sc Stabbable Grid Intersection Graphs Recognition (StabGIGRec)}}
{A bipartite graph $G$.}
{Question}
{Is there a  grid intersection representation of $G$ in which all segments are stabbed by the same straight line?}

We perform a reduction from the following problem. A \emph{Hamiltonian path} in a graph is a path that visits each vertex of the graph exactly once.

\Pb
{{\sc Planar Hamiltonian Path Completion Problem (PHPC)}}
{A planar graph $G$.}
{Question}
{Is $G$ a subgraph of a planar graph with a Hamiltonian path?}


\subsection{The reduction}\label{subsec:reduction2}

{\bf Convention.} Throughout our proof, we will assume that $G$ is a connected graph with at least five vertices, as every graph on four vertices is planar and the planarity of a non-connected graph is reduced to that of its components. 
\bigskip

Let $k$ be a fixed positive odd number. Given a planar graph $P$ on $p\geq 5$ vertices, we construct a bipartite apex graph $\apexgraph{P}$ in $\poly(p)$ time as follows. 

Let $\subdivide{P}$ be the full $k$-subdivision of $P$, \ie, $\subdivide{P}$ is the graph obtained by replacing each edge of $P$ with an induced path of length $k+1$. Formally, we replace each $e=xy\in E(P)$ with the path $(x, u^1_e, u^2_e, \ldots, u^k_e, y)$ (see Figure~\ref{fig:7sub}).
\begin{align*}
V(\subdivide{P}) &= V(P) \cup \{u^1_e, u^2_e, \ldots, u^k_e\mid e\in E(P)\}; \\
E(\subdivide{P}) &= \{x u^1_e, u^1_e u^2_e, \ldots, u^k_e y \mid e=xy\in E(P)\}.
\end{align*}

\begin{figure}[t]
    \centering
    	\scalebox{1}{

\rotatebox{90}{\begin{tikzpicture}
				
				\draw (0,4) -- (0,12);
				
				\def\A{0,15}
				\def\B{0,16}
				
				\def\C{0,4}
				\def\D{0,5}
				\def\E{0,6}
				\def\F{0,7}
				\def\G{0,8}
				\def\H{0,9}
				\def\I{0,10}
				\def\J{0,11}
				\def\K{0,12}

				\renewcommand{\vertexset}{(a,\A,blue!60!white,,,blue!60!white),(b,\B,blue!60!white,,,blue!60!white),(c,\C,blue!60!white,,,blue!60!white),(k,\K,blue!60!white,,,blue!60!white),(d,\D,magenta,,,magenta),(e,\E,magenta,,,magenta),(f,\F,magenta,,,magenta),(g,\G,magenta,,,magenta),(h,\H,magenta,,,magenta),(i,\I,magenta,,,magenta),(j,\J,magenta,,,magenta)}
				\renewcommand{\edgeset}{(a,b,,,0)}

				\renewcommand{\defradius}{0.07}
				
				\drawgraph
				
				
				\node[left] at (0,16) {\rotatebox{270} {\large $x$}};
				\node[left] at (0,15) {\rotatebox{270} {\large $y$}};
				
				\node[left] at (0,12) {\rotatebox{270} {\large $x$}};
				\node[left] at (0,11) {\rotatebox{270} {\large $u_e^1$}};
				\node[left] at (0,10) {\rotatebox{270} {\large $u_e^2$}};
				\node[left] at (0,9) {\rotatebox{270} {\large $u_e^3$}};
				\node[left] at (0,8) {\rotatebox{270} {\large $u_e^4$}};
				\node[left] at (0,7) {\rotatebox{270} {\large $u_e^5$}};
				\node[left] at (0,6) {\rotatebox{270} {\large $u_e^6$}};
				\node[left] at (0,5) {\rotatebox{270} {\large $u_e^7$}};
				\node[left] at (0,4) {\rotatebox{270} {\large $y$}};

			\end{tikzpicture}}}
    \caption{An edge $e=xy$ (left) and a 7-subdivision of $e$ (right).}
    \label{fig:7sub}
\end{figure}
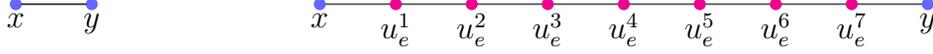

We call the vertices of $V(P)\subseteq V(\subdivide{P})$ as the \emph{original vertices} of $\subdivide{P}$ and the remaining vertices as the \emph{subdivision vertices} of $\subdivide{P}$.  Finally, we construct $\apexgraph{P}$ by adding a new vertex $a$ (called the `apex') to $\subdivide{P}$ and making it adjacent to all the original vertices of $\subdivide{P}$. Formally, $\apexgraph{P}$ is defined as follows.
\begin{align*}
V(\apexgraph{P}) &= V(\subdivide{P}) \cup \{a\}; \\
E(\apexgraph{P}) &= E(\subdivide{P}) \cup \{a v \mid v\in V(P)\}.
\end{align*}





The reduction from {\sc PHPC} to {\sc StabGIGRec} associates to the planar graph $G$, the input of {\sc PHPC}, the graph $\apexgraph{G}$ built as above, with $P=G$ and an arbitrary odd integer $k$ such that $k\geq 7$. 
See Figure~\ref{fig:representation} for an example of the reduction.

We need to prove that a planar graph $G$ is a yes-instance of {\sc PHPC} if and only if $\apexgraph{G}$ is a yes-instance of {\sc StabGIGRec}. The forward direction is Proposition~\ref{prop:forward} of the next section. For the backward direction, we will need the following result. A graph is a {\em 1-string graph} if it is the intersection graph of simple curves in the plane, also called {\em strings}, such that any two strings intersect at most once, and whenever they intersect they cross each other. 

\begin{proposition}[\cite{chakraborty2022}]\label{prop:yes-instance} 
For a planar graph $G$, if $k$ is odd and $\apexgraph{G}$ is a 1-string graph, then $G$ is a yes-instance of {\sc PHPC}.
\end{proposition}

\subsection{The proofs}\label{subsec:proofs2}

We start by the forward direction:

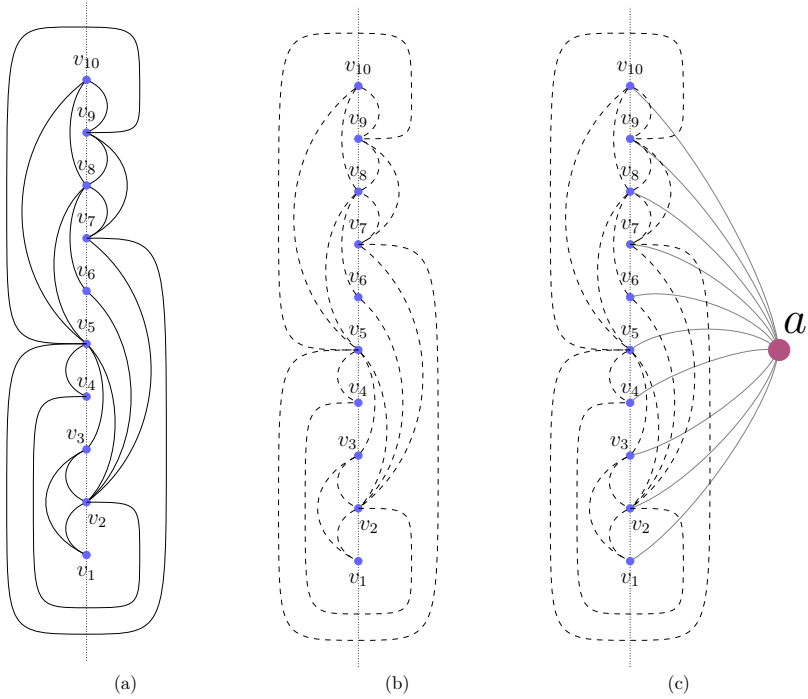
\begin{figure}[t]
	\centering
	\scalebox{0.7}{
		\begin{tabular}{ccc}
			\rotatebox{90}{
				\begin{tikzpicture}
				
				\draw[densely dotted] (-2,0) -- (10.5,0);
				\def\A{4,0}
				\def\B{3,0}
				\def\C{2,0}
				\def\D{1,0}
				\def\E{0,0}
				\def\F{5,0}
				\def\G{6,0}
				\def\H{7,0}
				\def\J{4,-3}
				\def\K{8,0}
				\def\L{9,0}
				\renewcommand{\vertexset}{(a,\A,blue!60!white,,,blue!60!white),(b,\B,blue!60!white,,,blue!60!white),(c,\C,blue!60!white,,,blue!60!white),(d,\D,blue!60!white,,,blue!60!white),(e,\E,blue!60!white,,,blue!60!white),(f,\F,blue!60!white,,,blue!60!white),(g,\G,blue!60!white,,,blue!60!white),(h,\H,blue!60!white,,,blue!60!white),(j,\J,white,,,white),(k,\K,blue!60!white,,,blue!60!white),(l,\L,blue!60!white,,,blue!60!white)}
				
				\renewcommand{\edgeset}{(a,b,,,-0.5),(c,d,,,-0.5),(d,e,,,-0.5),(c,e,,,-1),(a,c,,,0.4),(a,d,,,0.7),(f,d,,,1.1),(g,d,,,1.6),(h,f,,,-0.4),(h,a,,,-0.8),(l,h,,,-0.4),(k,g,,,1),(l,a,,,-1.6),(l,k,,,0.5),(h,k,,,-0.5),(g,h,,,-0.5)}
				
				\renewcommand{\defradius}{0.07}
				
				\drawgraph
				
				\draw (1,0) .. controls (1,-1) .. (0,-1) .. controls (-1,-1) .. (-1,0) .. controls (-1,1) .. (1,1) .. controls (3,1) .. (3,0);
				
				\draw (6,0) .. controls (6,-1.5) .. (2,-1.5) .. controls (-1.5,-1.5) .. (-1.5,0) .. controls (-1.5,1.5) .. (1,1.5) .. controls (4,1.5) .. (4,0);
				
				\draw (4,0) .. controls (4,1.5) .. (6.5,1.5) .. controls (10,1.5) .. (10,0.25) .. controls (10,-1) .. (8.75,-1) .. controls (8,-1) .. (8,0);
				
				\node[left] at (-0.1,0) {\rotatebox{270} {\large $v_1$}};
				\node[left] at (0.9,-0.2) {\rotatebox{270} {\large $v_2$}};
				\node[right] at (2,0.2) {\rotatebox{270} {\large $v_3$}};
				\node[right] at (3,0) {\rotatebox{270} {\large $v_4$}};
				\node[right] at (4.1,0) {\rotatebox{270} {\large $v_5$}};
				\node[right] at (5.1,0) {\rotatebox{270} {\large $v_6$}};
				\node[right] at (6.1,0) {\rotatebox{270} {\large $v_7$}};
				\node[right] at (7.1,0) {\rotatebox{270} {\large $v_8$}};
				\node[right] at (8.1,0) {\rotatebox{270} {\large $v_9$}};
				\node[right] at (9.1,0) {\rotatebox{270} {\large $v_{10}$}};
				\end{tikzpicture}} & 
    \rotatebox{90}{\begin{tikzpicture}
				
				\draw[densely dotted] (-2,0) -- (10.5,0);
				\def\A{4,0}
				\def\B{3,0}
				\def\C{2,0}
				\def\D{1,0}
				\def\E{0,0}
				\def\F{5,0}
				\def\G{6,0}
				\def\H{7,0}
				\def\J{4,-3}
				\def\K{8,0}
				\def\L{9,0}
				\renewcommand{\vertexset}{(a,\A,blue!60!white,,,blue!60!white),(b,\B,blue!60!white,,,blue!60!white),(c,\C,blue!60!white,,,blue!60!white),(d,\D,blue!60!white,,,blue!60!white),(e,\E,blue!60!white,,,blue!60!white),(f,\F,blue!60!white,,,blue!60!white),(g,\G,blue!60!white,,,blue!60!white),(h,\H,blue!60!white,,,blue!60!white),(j,\J,white,,,white),(k,\K,blue!60!white,,,blue!60!white),(l,\L,blue!60!white,,,blue!60!white)}
				
				\renewcommand{\edgeset}{(a,b,,,-0.5,dashed),(c,d,,,-0.5,dashed),(d,e,,,-0.5,dashed),(c,e,,,-1,dashed),(a,c,,,0.4,dashed),(a,d,,,0.7,dashed),(f,d,,,1.1,dashed),(g,d,,,1.6,dashed),(h,f,,,-0.4,dashed),(h,a,,,-0.8,dashed),(l,h,,,-0.4,dashed),(k,g,,,1,dashed),(l,a,,,-1.6,dashed),(l,k,,,0.5,dashed),(h,k,,,-0.5,dashed),(g,h,,,-0.5,dashed)}
				
				\renewcommand{\defradius}{0.07}

				\drawgraph
				
				\draw[dashed] (1,0) .. controls (1,-1) .. (0,-1) .. controls (-1,-1) .. (-1,0) .. controls (-1,1) .. (1,1) .. controls (3,1) .. (3,0);
				
				\draw[dashed] (6,0) .. controls (6,-1.5) .. (2,-1.5) .. controls (-1.5,-1.5) .. (-1.5,0) .. controls (-1.5,1.5) .. (1,1.5) .. controls (4,1.5) .. (4,0);
				
				\draw[dashed] (4,0) .. controls (4,1.5) .. (6.5,1.5) .. controls (10,1.5) .. (10,0.25) .. controls (10,-1) .. (8.75,-1) .. controls (8,-1) .. (8,0);
				
				\node[left] at (-0.1,0) {\rotatebox{270} {\large $v_1$}};
				\node[left] at (0.9,-0.2) {\rotatebox{270} {\large $v_2$}};
				\node[right] at (2,0.2) {\rotatebox{270} {\large $v_3$}};
				\node[right] at (3,0) {\rotatebox{270} {\large $v_4$}};
				\node[right] at (4.1,0) {\rotatebox{270} {\large $v_5$}};
				\node[right] at (5.1,0) {\rotatebox{270} {\large $v_6$}};
				\node[right] at (6.1,0) {\rotatebox{270} {\large $v_7$}};
				\node[right] at (7.1,0) {\rotatebox{270} {\large $v_8$}};
				\node[right] at (8.1,0) {\rotatebox{270} {\large $v_9$}};
				\node[right] at (9.1,0) {\rotatebox{270} {\large $v_{10}$}};

				\end{tikzpicture}}

                    & 
                    
                    \rotatebox{90}{\begin{tikzpicture}
				
				\draw[densely dotted] (-2,0) -- (10.5,0);
				\def\A{4,0}
				\def\B{3,0}
				\def\C{2,0}
				\def\D{1,0}
				\def\E{0,0}
				\def\F{5,0}
				\def\G{6,0}
				\def\H{7,0}
				\def\K{8,0}
				\def\L{9,0}
                \def\AP{4,-2.8}
				\renewcommand{\vertexset}{(ap,\AP,magenta!70!black,0.2,,magenta!70!black),(a,\A,blue!60!white,,,blue!60!white),(b,\B,blue!60!white,,,blue!60!white),(c,\C,blue!60!white,,,blue!60!white),(d,\D,blue!60!white,,,blue!60!white),(e,\E,blue!60!white,,,blue!60!white),(f,\F,blue!60!white,,,blue!60!white),(g,\G,blue!60!white,,,blue!60!white),(h,\H,blue!60!white,,,blue!60!white),(k,\K,blue!60!white,,,blue!60!white),(l,\L,blue!60!white,,,blue!60!white)}
				
				\renewcommand{\edgeset}{(ap,l,gray,,-0.5),(ap,k,gray,,-0.5),(ap,h,gray,,-0.5),(ap,g,gray,,-0.5),(ap,f,gray,,-0.5),(ap,a,gray,,-0.5),(ap,b,gray,,-0.3),(ap,c,gray,,0.35),(ap,d,gray,,0.5),(ap,e,gray,,0.5),(a,b,,,-0.5,dashed),(c,d,,,-0.5,dashed),(d,e,,,-0.5,dashed),(c,e,,,-1,dashed),(a,c,,,0.4,dashed),(a,d,,,0.7,dashed),(f,d,,,1.1,dashed),(g,d,,,1.6,dashed),(h,f,,,-0.4,dashed),(h,a,,,-0.8,dashed),(l,h,,,-0.4,dashed),(k,g,,,1,dashed),(l,a,,,-1.6,dashed),(l,k,,,0.5,dashed),(h,k,,,-0.5,dashed),(g,h,,,-0.5,dashed)}
				
				\renewcommand{\defradius}{0.07}

				\drawgraph
				
				\draw[dashed] (1,0) .. controls (1,-1) .. (0,-1) .. controls (-1,-1) .. (-1,0) .. controls (-1,1) .. (1,1) .. controls (3,1) .. (3,0);
				
				\draw[dashed] (6,0) .. controls (6,-1.5) .. (2,-1.5) .. controls (-1.5,-1.5) .. (-1.5,0) .. controls (-1.5,1.5) .. (1,1.5) .. controls (4,1.5) .. (4,0);
				
				\draw[dashed] (4,0) .. controls (4,1.5) .. (6.5,1.5) .. controls (10,1.5) .. (10,0.25) .. controls (10,-1) .. (8.75,-1) .. controls (8,-1) .. (8,0);
				
				\node[left] at (-0.1,0) {\rotatebox{270} {\large $v_1$}};
				\node[left] at (0.9,-0.2) {\rotatebox{270} {\large $v_2$}};
				\node[right] at (2,0.2) {\rotatebox{270} {\large $v_3$}};
				\node[right] at (3,0) {\rotatebox{270} {\large $v_4$}};
				\node[right] at (4.1,0) {\rotatebox{270} {\large $v_5$}};
				\node[right] at (5.1,0) {\rotatebox{270} {\large $v_6$}};
				\node[right] at (6.1,0) {\rotatebox{270} {\large $v_7$}};
				\node[right] at (7.1,0) {\rotatebox{270} {\large $v_8$}};
				\node[right] at (8.1,0) {\rotatebox{270} {\large $v_9$}};
				\node[right] at (9.1,0) {\rotatebox{270} {\large $v_{10}$}};
				
				\node[right] at (4.2,-3.1) {\rotatebox{270} {\Huge $a$}};

				\end{tikzpicture}}\\
			(a) & (b) & (c)
   
	\end{tabular}}
	\caption{
 (a) An example graph $G$ which is a yes-instance of {\sc PHPC}. (b) Each edge of the graph is replaced with a path with $k+1$ edges. (A dashed curve indicates a path of length $k+1$.) (c) The graph $\apexgraph{G}$, constructed by adding one extra vertex (the apex vertex $a$) in $\subdivide{G}$. }\label{fig:representation}
\end{figure}

\begin{proposition}\label{prop:forward}
If $G$ is a yes-instance of {\sc PHPC}, then $\apexgraph{G}$ is a $\StabGIG$. 
\end{proposition}

The proof of this result is presented in Sections~\ref{subsec:overview} to \ref{subsec:Gapexrepres}. Section~\ref{subsec:proofthm} contains the proof of
Theorem~\ref{thm:stabbable-GIG}.



We use the notation $(x,y)$, where $x$ and $y$ are real values, for the Cartesian coordinates of a point, and the notation $[(x_1,y_1),(x_2,y_2)]$ for the segment whose endpoints are  $(x_1,y_1)$ and $(x_2,y_2)$.

\subsubsection{Overview of the construction}\label{subsec:overview}
In view of proving Proposition \ref{prop:forward}, let $G$ be a yes-instance of PHPC, and let $V(G)=\{v_1,v_2,\ldots,v_n\}$. We start with a plane drawing of $G$, \ie, a drawing of $G$ in the plane, in which the vertices are points and the edges are strings (that is, simple curves). Then we modify the drawing in a step-by-step manner to end with a $\StabGIG$ representation of $\apexgraph{G}$. Let us now provide an overview of these steps.

\begin{enumerate}[label=(\alph*)]

    \item \label{enum:step1auergleisner} Consider a drawing of $G$ in the plane in which the vertex $v_i$ (for every $i\in[n]$) has the Cartesian coordinate $(0,i-1)$, each edge of $G$ crosses the y-axis (the line $x=0)$ at most once, and no edge of $G$ crosses the y-axis between $y=0$ and $y=n-1$. (Auer \& Glei{\ss}ner~\cite{auer2011} show that such a drawing exists for every graph that is a yes-instance of {\sc PHPC}.) See Figure~\ref{fig:representation}(a) for an illustration, where $v_1$ is located at $(0,0)$ and $v_{10}$ is located at $(0,9)$.
    
    \item \label{enum:step2creationofH} In this drawing, though no edge of $G$ crosses the y-axis between the points $(0,0)$ and $(0,n-1))$, there might be edges of $G$ that cross the y-axis below $y=0$ or above $y=n-1$. We place new vertices on the intersection points of those edges with the y-axis, to create a new graph $H$. Let $n_H$ be the number of vertices in $H$. See Figure~\ref{fig:proof-illustration} for an illustration, where there are three such vertices, indicated by square nodes.
    
    \item \label{enum:step3propertiesofH} The graph $H$ has some nice properties. Firstly, $n\leq n_H\leq 4n$. Secondly, $H$ is a subdivision of $G$. As subdivision does not affect planarity, $H$ is also planar. Thirdly, the drawing of $H$ ensures that no edge of $H$ crosses the y-axis. We fix the line $y=x$ as our stab line $\ell$ (recall that $\ell$ does need not to be axis-parallel), and use the current plane drawing of $H$ to devise a $\StabGIG$ representation of the graph $\apexgraphthree{H}$. (The graph $\apexgraphthree{H}$ is obtained as described in the previous section, by considering $P=H$ and $k=3$. See Figure~\ref{fig:proof-illustration}.) Note that the graph $\apexgraph{G}$ is obtained from $\apexgraphthree{H}$ by further subdividing some of its edges. 
    
    \item \label{enum:step4representationofH} In the $\StabGIG$ representation of $\apexgraphthree{H}$, the apex segment of $\apexgraphthree{H}$ is represented by a long vertical segment coinciding with the y-axis. The $n_H$ original vertices of $\apexgraphthree{H}$ are represented by $n_H$ horizontal segments, each one crossing the y-axis at the point defined by its corresponding vertex in the plane representation of $H$. Each edge $e=xy$ of $H$ (which is a string in the plane drawing of $H$) becomes a path with three intermediate vertices $u_e^1, u_e^2, u_e^3$ in $\apexgraphthree{H}$. These three vertices are represented as three segments that form a $\U$-shape or a $\Urev$-shape, whose two vertical segments appropriately intersect the horizontal segments representing respectively $x$ and $y$. See Figure~\ref{fig:FigEdgeRep} for an illustration (formal details are given below). 

    \item For each edge $e$ of $H$, we need to very precisely and carefully describe the positions of the three segments corresponding to its three subdivision vertices (which will form a $\U$-shape or a $\Urev$-shape) $u^1_e, u^2_e, u^3_e$ in $\apexgraphthree{H}$, in such a way that all the three segments intersect the stab line $\ell$. This is the most non-trivial part of our proof. See Figure~\ref{fig:FigEdgeRep} for an illustration. \label{enum:step5nontrivial} 
    
    \item \label{enum:step6sevensubdiv} Lastly, we need to convert $\apexgraphthree{H}$ to $\apexgraph{G}$. Note that $\apexgraphthree{H}\setminus\{a\}$ is a subdivision of $G$ in which each edge of $G$ is subdivided either 3 or 7 times. 
    An overview of this construction is illustrated by Figure~\ref{fig:LastSubdivision}. Finally, with the apex segment $[(0,0),(0,n-1)]$, we have a $\StabGIG$ representation of $\apexgraph{G}$.
    
\end{enumerate}

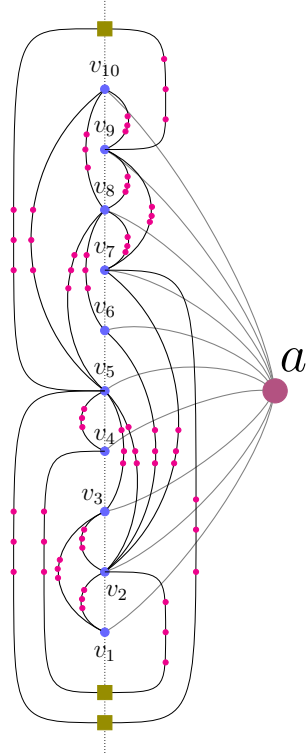
\begin{figure}[t]
    \centering
    	\scalebox{0.8}{

\rotatebox{90}{\begin{tikzpicture}
				
				\draw[densely dotted] (-2,0) -- (10.5,0);
				\def\A{4,0}
				\def\B{3,0}
				\def\C{2,0}
				\def\D{1,0}
				\def\E{0,0}
				\def\F{5,0}
				\def\G{6,0}
				\def\H{7,0}
				\def\K{8,0}
				\def\L{9,0}
                    
                \def\AP{4,-2.8}
                \renewcommand{\vertexset}{(ap,\AP,magenta!70!black,0.2,,magenta!70!black),(a,\A,blue!60!white,,,blue!60!white),(b,\B,blue!60!white,,,blue!60!white),(c,\C,blue!60!white,,,blue!60!white),(d,\D,blue!60!white,,,blue!60!white),(e,\E,blue!60!white,,,blue!60!white),(f,\F,blue!60!white,,,blue!60!white),(g,\G,blue!60!white,,,blue!60!white),(h,\H,blue!60!white,,,blue!60!white),(k,\K,blue!60!white,,,blue!60!white),(l,\L,blue!60!white,,,blue!60!white)}

                \renewcommand{\edgeset}{(ap,l,gray,,-0.5),(ap,k,gray,,-0.5),(ap,h,gray,,-0.5),(ap,g,gray,,-0.5),(ap,f,gray,,-0.5),(ap,a,gray,,-0.5),(ap,b,gray,,-0.3),(ap,c,gray,,0.35),(ap,d,gray,,0.5),(ap,e,gray,,0.5),(a,b,,,-0.5),(c,d,,,-0.5),(d,e,,,-0.5),(c,e,,,-1),(a,c,,,0.4),(a,d,,,0.7),(f,d,,,1.1),(g,d,,,1.6),(h,f,,,-0.4),(h,a,,,-0.8),(l,h,,,-0.4),(k,g,,,1),(l,a,,,-1.6),(l,k,,,0.5),(h,k,,,-0.5),(g,h,,,-0.5)}
				
				\renewcommand{\defradius}{0.07}

				\drawgraph
				
				\draw (1,0) .. controls (1,-1) .. (0,-1) .. controls (-1,-1) .. (-1,0) .. controls (-1,1) .. (1,1) .. controls (3,1) .. (3,0);
				
				\draw (6,0) .. controls (6,-1.5) .. (2,-1.5) .. controls (-1.5,-1.5) .. (-1.5,0) .. controls (-1.5,1.5) .. (1,1.5) .. controls (4,1.5) .. (4,0);
				
				\draw (4,0) .. controls (4,1.5) .. (6.5,1.5) .. controls (10,1.5) .. (10,0.25) .. controls (10,-1) .. (8.75,-1) .. controls (8,-1) .. (8,0);
				
				\node[left] at (-0.1,0) {\rotatebox{270} {\large $v_1$}};
				\node[left] at (0.9,-0.2) {\rotatebox{270} {\large $v_2$}};
				\node[right] at (2,0.2) {\rotatebox{270} {\large $v_3$}};
				\node[right] at (3,0) {\rotatebox{270} {\large $v_4$}};
				\node[right] at (4.1,0) {\rotatebox{270} {\large $v_5$}};
				\node[right] at (5.1,0) {\rotatebox{270} {\large $v_6$}};
				\node[right] at (6.1,0) {\rotatebox{270} {\large $v_7$}};
				\node[right] at (7.1,0) {\rotatebox{270} {\large $v_8$}};
				\node[right] at (8.1,0) {\rotatebox{270} {\large $v_9$}};
				\node[right] at (9.1,0) {\rotatebox{270} {\large $v_{10}$}};

                \node[rectangle,draw,olive,fill=olive] at (10,0) {};
                \node[rectangle,draw,olive,fill=olive] at (-1,0) {};
                \node[rectangle,draw,olive,fill=olive] at (-1.5,0) {};
                
                \filldraw[magenta] (2,1) circle (1.2pt); \filldraw[magenta] (1.5,1) circle (1.2pt); \filldraw[magenta] (1,1) circle (1.2pt);

                \filldraw[magenta] (2,1.5) circle (1.2pt); \filldraw[magenta] (1.5,1.5) circle (1.2pt); \filldraw[magenta] (1,1.5) circle (1.2pt);

                \filldraw[magenta] (7,1.5) circle (1.2pt); \filldraw[magenta] (6.5,1.5) circle (1.2pt); \filldraw[magenta] (6,1.5) circle (1.2pt);

                \filldraw[magenta] (7,1.18) circle (1.2pt); \filldraw[magenta] (6.5,1.21) circle (1.2pt); \filldraw[magenta] (6,1.18) circle (1.2pt);

                \filldraw[magenta] (2.2,-1.5) circle (1.2pt); \filldraw[magenta] (1.7,-1.5) circle (1.2pt); \filldraw[magenta] (1,-1.5) circle (1.2pt);

                \filldraw[magenta] (0.5,-1) circle (1.2pt); \filldraw[magenta] (0,-1) circle (1.2pt); \filldraw[magenta] (-0.5,-1) circle (1.2pt);

                \filldraw[magenta] (9.5,-0.975) circle (1.2pt); \filldraw[magenta] (9,-1) circle (1.2pt); \filldraw[magenta] (8.5,-1) circle (1.2pt);

                \filldraw[magenta] (8.3,-0.34) circle (1.2pt); \filldraw[magenta] (8.4,-0.37) circle (1.2pt); \filldraw[magenta] (8.55,-0.375) circle (1.2pt);

                \filldraw[magenta] (3,-0.315) circle (1.2pt); \filldraw[magenta] (2.8,-0.3) circle (1.2pt); \filldraw[magenta] (3.35,-0.275) circle (1.2pt);

                \filldraw[magenta] (3,-0.827) circle (1.2pt); \filldraw[magenta] (2.8,-0.825) circle (1.2pt); \filldraw[magenta] (3.35,-0.825) circle (1.2pt);

                \filldraw[magenta] (3,-1.17) circle (1.2pt); \filldraw[magenta] (2.8,-1.14) circle (1.2pt); \filldraw[magenta] (3.35,-1.21) circle (1.2pt);

                \filldraw[magenta] (3,-0.5) circle (1.2pt); \filldraw[magenta] (2.8,-0.525) circle (1.2pt); \filldraw[magenta] (3.4,-0.39) circle (1.2pt);
                
                \filldraw[magenta] (8,0.32) circle (1.2pt); \filldraw[magenta] (7.7,0.28) circle (1.2pt); \filldraw[magenta] (8.25,0.3) circle (1.2pt);

                \filldraw[magenta] (6,0.32) circle (1.2pt); \filldraw[magenta] (5.7,0.28) circle (1.2pt); \filldraw[magenta] (6.25,0.3) circle (1.2pt);

                \filldraw[magenta] (6,0.555) circle (1.2pt); \filldraw[magenta] (5.7,0.6) circle (1.2pt); \filldraw[magenta] (6.25,0.5) circle (1.2pt);
                
                \filldraw[magenta] (7.3,-0.34) circle (1.2pt); \filldraw[magenta] (7.4,-0.37) circle (1.2pt); \filldraw[magenta] (7.55,-0.375) circle (1.2pt);

                \filldraw[magenta] (6.8,-0.74) circle (1.2pt); \filldraw[magenta] (6.9,-0.77) circle (1.2pt); \filldraw[magenta] (7.05,-0.775) circle (1.2pt);
                
                \filldraw[magenta] (6.3,-0.34) circle (1.2pt); \filldraw[magenta] (6.4,-0.37) circle (1.2pt); \filldraw[magenta] (6.55,-0.375) circle (1.2pt);

                \filldraw[magenta] (0.7,0.34) circle (1.2pt); \filldraw[magenta] (0.4,0.37) circle (1.2pt); \filldraw[magenta] (0.55,0.375) circle (1.2pt);

                \filldraw[magenta] (1.2,0.74) circle (1.2pt); \filldraw[magenta] (0.9,0.77) circle (1.2pt); \filldraw[magenta] (1.05,0.775) circle (1.2pt);

                \filldraw[magenta] (1.7,0.34) circle (1.2pt); \filldraw[magenta] (1.4,0.37) circle (1.2pt); \filldraw[magenta] (1.55,0.375) circle (1.2pt);

                \filldraw[magenta] (3.7,0.34) circle (1.2pt); \filldraw[magenta] (3.4,0.37) circle (1.2pt); \filldraw[magenta] (3.55,0.375) circle (1.2pt);
                
                \node[right] at (4.2,-3.1) {\rotatebox{270} {\Huge $a$}};
                
			\end{tikzpicture}}}
    \caption{The graph $\apexgraphthree{H}$, constructed from the the graph $G$ shown in Figure~\ref{fig:representation}(a) using the intermediate graph $H$. Note that $H$ is obtained from $G$ by introducing three new vertices (indicated by square nodes) that perform a first subdivision of the three edges that cross the y-axis, namely $v_2v_4, v_5v_7$ and $v_5v_9$. The graph $\apexgraph{G}$ is then obtained from $\apexgraphthree{H}$ by introducing an even number of additional subdivision vertices for all edges $v_iv_j$ of $G$ that are insufficiently subdivided.}
    \label{fig:proof-illustration}
\end{figure}

\subsubsection{Construction of $\apexgraphthree{H}$ and its $\StabGIG$ representation}

Now we commence with the formal proof of Proposition~\ref{prop:forward}. Let us start with~\autoref{enum:step1auergleisner}. Since $G$ is a yes-instance of {\sc PHPC}, according to Auer \& Glei{\ss}ner~\cite{auer2011}, the edge set $E(G)$ can be partitioned into three sets $\Eleftedge(G)$, $\Erightedge(G)$, and $\Ecrossedge(G)$, such that $G$ admits a plane drawing satisfying the following properties.
\begin{itemize}
    \item For every $i\in[n]$, the vertex $v_i$ has the Cartesian coordinate $(0,i-1)$.
    \item Every edge $e\in\Eleftedge(G)$ lies in the half-plane $x\leq 0$.
    \item Every edge $e\in\Erightedge(G)$ lies in the half-plane $x\geq 0$.
    \item Every edge $e\in\Ecrossedge(G)$ crosses the line $x=0$ exactly once, at $(0,Y_e)$, where $Y_e<0$ or $Y_e>n-1$.
\end{itemize}
\autoref{fig:representation}(a) presents an example of a graph, and of a plane drawing of it which satisfies the above properties.  Note that the above plane drawing crucially uses the fact that the line $x=0$ partitions the plane into two half-planes ($x\leq 0$ and $x\geq 0$). Now, let us move on to~\autoref{enum:step2creationofH}.

The construction of $H$ from $G$ is done as follows. For each edge $e=v_iv_j\in\Ecrossedge(G)$, we add a new vertex $\newvert_e$ to the graph, represented by the point $(0,Y_e)$. Consequently, the edge $e=v_iv_j$ is replaced by two edges $v_i\newvert_e$ and $\newvert_e,v_j$. One of those edges is in $\Eleftedge(H)$ and the other is in $\Erightedge(H)$. Therefore, the vertex $\newvert_e$ is a subdivision vertex of the edge $e$. More formally,
\begin{align*}
    V(H)&=V(G)\cup\{\newvert_e\mid e\in\Ecrossedge(G)\};\\
    E(H)&=\Eleftedge(G)\cup\Erightedge(G)\cup\{x\newvert_e,\newvert_e y\mid e=xy\in\Ecrossedge(G)\}.
\end{align*}
What did this achieve for us? The graph $H$ has some nice properties, which we mentioned in~\autoref{enum:step3propertiesofH}, and will elaborate now.
\begin{itemize}
    \item $H$ is a subdivision of $G$. Edges of $\Eleftedge(G)$ and $\Erightedge(G)$ are retained as-is in $H$, and edges of $\Ecrossedge(G)$ are subdivided once.
    \item $H$ is a planar graph. This is because $G$ is a planar graph, and subdivision does not affect planarity.
    \item No edge of $H$ crosses the y-axis ($\Ecrossedge(H)=\emptyset$). This is a straightforward feature of our construction.
    \item $n\leq n_H \leq 4n$, where $n_H=|V(H)|$. This is because $G$ has at most $3n$ edges (as $G$ is planar), and each edge of $G$ is subdivided at most once in $H$. Since $H$ also retains the vertices of $G$, $n_H \leq n+3n$.
\end{itemize}
From $H$, we construct $\apexgraphthree{H}$, as described in Section~\ref{subsec:reduction2}. Then we progressively build its $\StabGIG$ representation (\autoref{enum:step4representationofH}). We first define the positions of the segments corresponding to the apex vertex of $\apexgraphthree{H}$ and the original vertices of $H$ in our $\StabGIG$ representation of $\apexgraphthree{H}$. Note that the apex vertex in $\apexgraphthree{H}$ is called $a$, just as the apex vertex in $\apexgraph{G}$, since these vertices finally coincide in our construction. 

The vertices with the minimum and maximum y-coordinates in $H$ are at $(0,Y_{\min})$ and $(0,Y_{\max})$, where
\begin{align*}
Y_{\min}&=\operatorname{min}\left(\{0\}\cup\{Y_e\mid e\in\Ecrossedge(G)\}\right);\\
Y_{\max}&=\operatorname{max}\left(\{n-1\}\cup\{Y_e\mid e\in\Ecrossedge(G)\}\right).
\end{align*}
We may assume that $Y_{\min}=0$ and $Y_{\max}=n_H-1$, and that the other vertices of $H$ can be shifted accordingly so that their coordinates are $(0,0), (0,1), (0,2), \ldots, (0,n_H-1)$.\footnote{This can be done, for example, by shifting the x-axis downwards and making the gap between consecutive points as $1$.} Furthermore, the vertices are suitably renamed as $w_1, w_2, w_3, \ldots, w_{n_H}$ so that each $w_i$ has the coordinate $(0,i-1)$. 
For each vertex $x\in V(\apexgraphthree{H})$, let $\curve{x}$ be its corresponding segment. See Figure~\ref{fig:FigEdgeRep}.
\begin{align*}
    &\text{The apex vertex $a$ of $\apexgraphthree{H}$:}&&\curve{a} = [(0,0),(0,n_H-1)]\\
    &\text{The original vertices of $H$:}&&\curve{w_i} = [(-i-0.1,i-1),(2n_H-i+0.1,i-1)]\qquad\forall\,i\in[n_H].
\end{align*}

 \begin{figure}[t!]
 \scalebox{0.42}{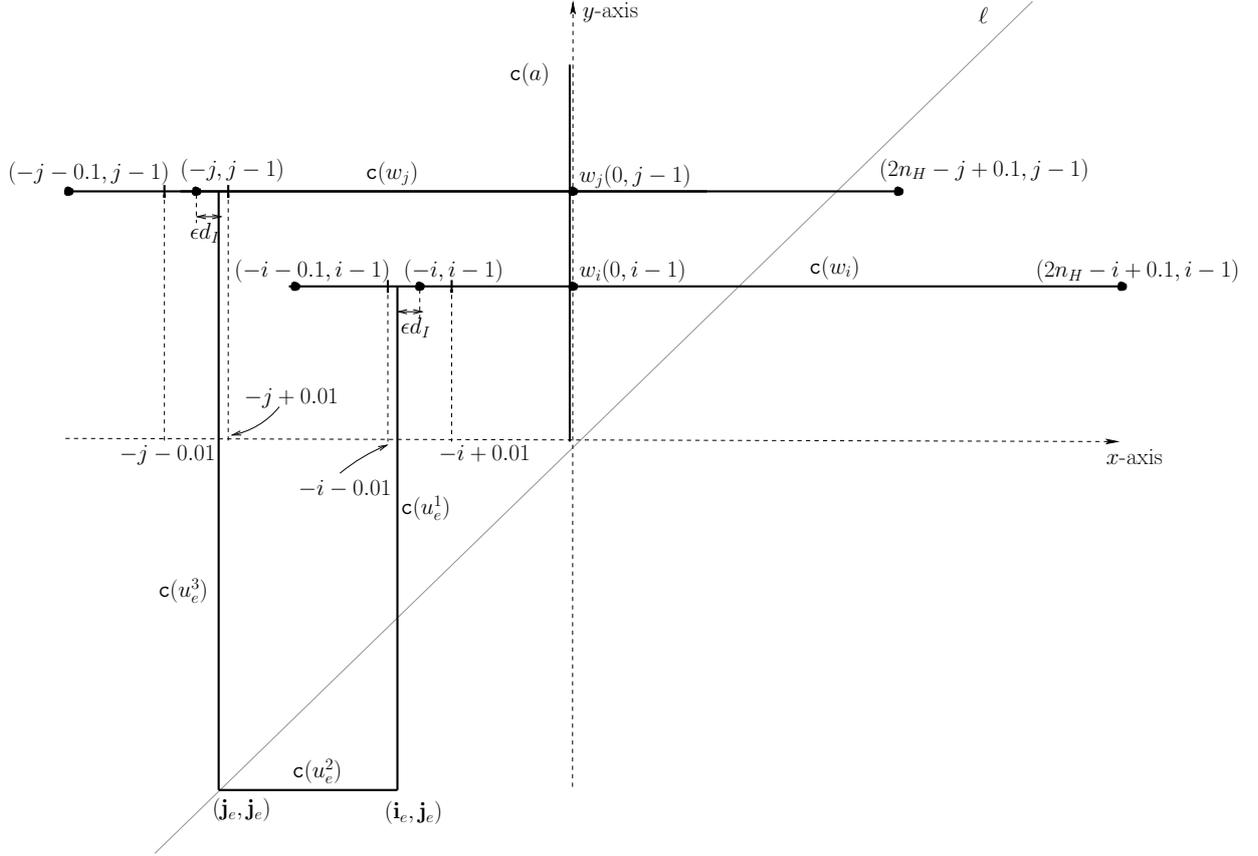}
 \caption{$\StabGIG$ representation of the path $(w_i, u_e^1, u_e^2, u_e^3,w_j)$ of $\apexgraphthree{H}$ corresponding to the edge $e=w_iw_j$ of $H$, with $i<j$. 
 }
 \label{fig:FigEdgeRep}
 \end{figure}

We fix the line $y=x$ as our stab line $\ell$ (recall that $\ell$ does not need to be axis-parallel). It is easy to see that all the segments described so far intersect the stab line. 
This brings us to~\autoref{enum:step5nontrivial}, where we need to place the segments corresponding to the three subdivision vertices of $e$ for each edge $e$ of $H$.

Let us consider edges of $\Eleftedge(H)$. Later, we will deal with edges of $\Erightedge(H)$. We construct a set of intervals $\Ileft$ from $\Eleftedge(H)$ as follows. For each edge $e=w_iw_j\in\Eleftedge(H)$ (where $i<j$), the set $\Ileft$ contains an open interval $(i-1,j-1)$. Thus, the number of intervals in $\Ileft$ is the number of edges in $\Eleftedge(H)$. Now, note that $\Ileft$ has the following interesting property.

\bigskip

\begin{observation} \label{obs:bookthickness}
    If two intervals of $\Ileft$ have an overlap (\ie, the intervals are not disjoint), then one of them must be contained in the other.
\end{observation}

\bigskip

Observation~\ref{obs:bookthickness} is well known and easy to see; it is the basic underlying principle of \emph{book thickness} or \emph{book embeddings} of graphs (for a proof, see~\cite{bernhartkainen1979}). 
We will use this observation crucially in our construction of the three subdivision segments for each edge $e$ of $\Eleftedge(H)$.

Given an interval $I\in\Ileft$, let its \emph{depth}, denoted by $\level{I}$, be the number of intervals $I'\in\Eleftedge(H)$ such that $I\subseteq I'$. In other words, $\level{I}$ is the number of intervals that contain $I$ (including $I$ itself). Thus, $1\leq\level{I}\leq n_H-1$. Now we are ready to describe our construction.

Fix $\varepsilon=0.01/n_H$. Let $e=w_iw_j\in\Eleftedge(H)$ such that $i<j$. Let $I=(i-1,j-1)$ be the interval corresponding to $e$ in $\Ileft$. After the $3$-subdivision, the edge $e$ becomes the path $(w_i,u_e^1,u_e^2,u_e^2,w_j)$ in $\apexgraphthree{H}$. We already know the coordinates of the segments for $w_i$ and $w_j$. As for the other three segments, we define $\lefty_e$ and $\righty_e$ as follows.
\begin{align*}
    \lefty_e &= -j+\varepsilon\level{I};\\
    \righty_e &= -i-\varepsilon\level{I}.
\end{align*}
Now, we are set to define the coordinates of the other three subdivision segments of $e$, which draw a $\U$-shape. See Figure~\ref{fig:FigEdgeRep}.
\begin{align*}
    \curve{u_e^3} &= \left[(\lefty_e,j-1),(\lefty_e,\lefty_e)\right];\\
    \curve{u_e^2} &= \left[(\lefty_e,\lefty_e),(\righty_e,\lefty_e)\right];\\
    \curve{u_e^1} &= \left[(\righty_e,\lefty_e),(\righty_e,i-1)\right].
\end{align*}
Before proving that the 3 segments do not intersect any segments apart from $w_i$, $w_j$, and each other, let us prove that all 3 segments intersect the stab line. Consider the point of intersection of $\curve{u_e^3}$ and $\curve{u_e^2}$, namely $(\lefty_e,\lefty_e)$. This point has the same x- and y-coordinate, and so it lies on the stab line $y=x$. Thus, all that remains to be shown is that $\curve{u_e^1}$ also intersects the stab line. Since $i\leq j-1$ and $\varepsilon\level{I}<0.01$, $$\lefty_e \leq \righty_e \leq i-1.$$ Hence, the point $(\righty_e,\righty_e)$ lies on the segment $\curve{u_e^1} = \left((\righty_e,\lefty_e),(\righty_e,i-1)\right)$, and on the stab line $y=x$ also.

Now, we will show that for two edges $e=w_iw_j$ (where $i<j$) and $e'=w_{i'}w_{j'}$ (where $i'<j'$) of $\Eleftedge(H)$, the segments corresponding to the three subdivision vertices of $e$ do not intersect the segments corresponding to the three subdivision vertices of $e'$. Let their corresponding intervals in $\Ileft$ be $I$ and $I'$, respectively. Due to Observation~\ref{obs:bookthickness}, we have only two cases.
\begin{enumerate}
    \item \underline{Case 1: $I$ and $I'$ are disjoint.} Suppose that $i<j\leq i'<j'$ (the other sub-case, $i'<j'\leq i<j$, is similar). We have the following.
    \begin{align*}
        i < j&\leq i' < j'\\
        -j' < -i'&\leq -j < -i\\
        -j'+\varepsilon\level{I'}<-i'-\varepsilon\level{I'}&< -j+\varepsilon\level{I}<-i-\varepsilon\level{I} &&\text{(since $2\varepsilon\level{I'}<1\leq j'-i'$ and $2\varepsilon\level{I}<1\leq j-i$)}\\
        \lefty_{e'} < \righty_{e'} &< \lefty_{e} < \righty_{e}
    \end{align*}
    
    Thus, the x-coordinate of each of the 3 subdivision segments of $e'$ is at most $\righty_{e'}$, and the x-coordinate of each of the 3 subdivision segments of $e$ is at least $\lefty_e$. This completes the proof of Case 1.
    
    \item \underline{Case 2: One of $I$ or $I'$ contains the other.} Suppose that $I\subseteq I'$ (the other sub-case, $I'\subseteq I$, is similar). Then, note that $\level{I'}<\level{I}$. We have the following.
    \begin{align*}
        i' \leq i &< j \leq j' &&\text{(since $I\subseteq I'$)}\\
        -j' \leq -j &< -i \leq -i'\\
        -j' + \varepsilon\level{I'} < -j + \varepsilon\level{I} &< -i - \varepsilon\level{I} < -i' - \varepsilon\level{I'} &&\text{(since $\level{I'}<\level{I}$ and $2\varepsilon\level{I}<1\leq j-i$)}\\
        \lefty_{e'}<\lefty_{e}&<\righty_{e}<\righty_{e'}
    \end{align*}
    Thus, the x-coordinate of $\curve{u_{e'}^3}$ (namely $\lefty_{e'}$) is less than the x-coordinate of each of the 3 subdivision segments of $e$, and the x-coordinate of $\curve{u_{e'}^1}$ (namely $\righty_{e'}$) is greater than the x-coordinate of each of the 3 subdivision segments of $e$. Finally, the y-coordinate of $\curve{u_{e'}^2}$ (namely $\lefty_{e'}$) is less than the y-coordinate of each of the 3 subdivision segments of $e$. This completes the proof of Case 2.
\end{enumerate}

This completes the construction of the $\StabGIG$ representation for the subgraph of $\apexgraphthree{H}$ induced by the original vertices of $H$, the apex vertex $a$ and the subdivision vertices located in the left half-plane. For the vertices and edges located in the right half-plane, the construction is similar but replaces the $\U$-shapes by $\Urev$-shapes. The correctness of our construction follows from Cases 1 and 2 above.

\subsubsection{$\StabGIG$ representation of $\apexgraph{G}$}\label{subsec:Gapexrepres}

 In order to conclude our proof of Proposition~\ref{prop:forward}, we need -- according to~\autoref{enum:step6sevensubdiv} -- to convert the $\StabGIG$ representation of $\apexgraphthree{H}$ into a $\StabGIG$ representation of $\apexgraph{G}$. 
This operation reduces to further subdividing one edge of $\apexgraphthree{H}$, that may be appropriately chosen, by transforming it into a path of odd length $2h+1$, for an appropriate integer $h\geq 1$. We always choose an edge represented by a horizontal segment, and subdivide it (that is, we replace it with a sequence of segments), as illustrated in Figure~\ref{fig:LastSubdivision}.

Consider the $\StabGIG$ representation of $\apexgraphthree{H}$. For each edge $e\in \Ecrossedge(G)$, a 7-subdivision is performed in  $\apexgraphthree{H}$. But for each edge $e\in\Eleftedge(G)\cup\Erightedge(G)$, only a 3-subdivision is performed in $\apexgraphthree{H}$. When $k=7$, the latter edges are insufficiently subdivided in $\apexgraphthree{H}$ with respect to $\apexgraph{G}$. Moreover, when $k>7$ all the original edges in $G$ are insufficiently subdivided. In all cases, let $e\in E(G)$ be one of the edges that are insufficiently subdivided, and assume that $e\in\Eleftedge(H)$. (Otherwise, when $e\in\Erightedge(H)$ the approach is similar, whereas when $e\in\Ecrossedge(G)$ we have to perform the same approach on the 'half-edge' of $e$ that belongs to $\Eleftedge(H)$.) In  $\apexgraphthree{H}$, $e$ is 3-subdivided. Assume we need a $k$-subdivision of $e$, with odd $k\geq 7$, and let $h=(k-3)/2$. Then we perform a $2h$-subdivision of the edge $u_e^1u_e^2$ of $\apexgraphthree{H}$,
and describe below the new segments representing $u_e^1, u_e^2$ and the $2h$ supplementary vertices. 

In $\apexgraphthree{H}$, the intersection point between the segments  $\curve{u_e^3}$ and $\curve{u_e^2}$ is $(\lefty_e,\lefty_e)$. The closest possible similar point located on the stab line towards right is the point $(\lefty_e+\varepsilon,\lefty_e+\varepsilon)$. We then replace the segment $\curve{u_e^2}$ with a succession of $2h+1$ alternating horizontal and vertical segments whose intersections define the path required by the $2h$-subdivision of 
$(u_e^1,u_e^2)$. See Figure~\ref{fig:LastSubdivision}.

 \begin{figure}[t!]
 \centering
 \includegraphics[scale=0.5]{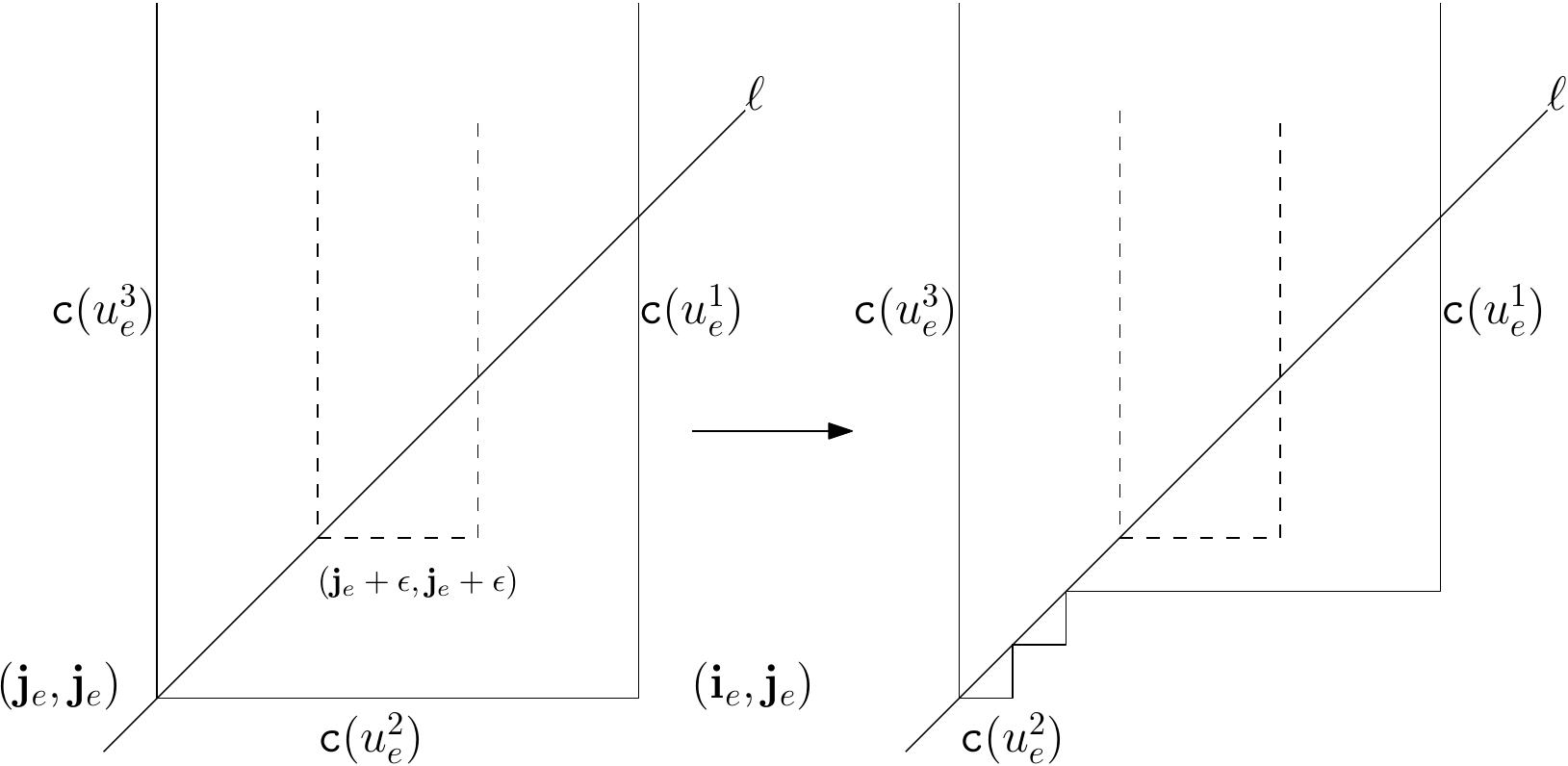}
 \caption{Replacement of $\curve{u_e^2}$ (left), for some edge $e$ which is not sufficiently subdivided in $\apexgraphthree{H}$, with a sequence of segments.
 The dotted $\U$-shape indicates the closest possible location of an $\U$-shape located above the one we defined for $e$. 
 }
 \label{fig:LastSubdivision}
 \end{figure}

Formally, let $\varepsilon_2=\varepsilon/(h+1)$, and define the following segments. The horizontal segment obtained for $t=0$ is the new segment $\curve{u_e^2}$. 


For $0\leq t\leq h-1$:

$$\left[(\lefty_e+t\varepsilon_2,\lefty_e+t\varepsilon_2), (\lefty_e+(t+1)\varepsilon_2,\lefty_e+t\varepsilon_2)\right]$$
$$\left[(\lefty_e+(t+1)\varepsilon_2,\lefty_e+t\varepsilon_2), (\lefty_e+(t+1)\varepsilon_2,\lefty_e+(t+1)\varepsilon_2)\right]$$








\newpage

Complete this set of segments with:

$$\left[(\lefty_e+h\varepsilon_2,\lefty_e+h\varepsilon_2), (\righty_e,\lefty_e+h\varepsilon_2)\right]$$
$$\curve{u_e^1}=\left[(\righty_e,\lefty_e+h\varepsilon_2), (\righty_e,i-1)\right]$$

All the segments corresponding to the $2h$ new subdivision vertices lie strictly below the closest upper segment not involved in the subdivision of $e$ (whose $y$-coordinate is at least $\lefty_e+\varepsilon$), and (not strictly) above the former segment $\curve{u_e^2}$. They also lie  between the $x$-coordinates of $\curve{u_e^3}$ and $\curve{u_e^1}$. Thus they cannot intersect other segments. So we have a representation of the path with $(2h+1)$ edges required by the subdivision.

This ends the proof of Proposition \ref{prop:forward}.

\subsubsection{Proof of Theorem \ref{thm:stabbable-GIG}}\label{subsec:proofthm}
We argue that recognizing $\StabGIG$ is in $\NP$. Note that a graph is a $\StabGIG$ if and only if it has a stabbable grid intersection representation where the stab line has slope $1$, and the intersection points of the segments with the stab line as well as with the other segments have integer coordinates. Such a representation is called \emph{good}. The backward direction is immediate.  The forward direction is proved as follows. Let $G$ be a {\sc StabGIG}, and consider an arbitrary stabbed grid intersection representation of it. First, appropriately rotate and then scale the representation in the horizontal direction, in order to obtain a grounded line with slope $1$. Second, order the segments according to the bottom to top order of their intersection points with the stab line. Given a segment $s$, let $r(s)$ be its rank in this order. Third, build a new stabbed representation, as follows. For each horizontal segment $s$ of the initial representation, its  new $y$-coordinate
is $r(s)$, and the new $x$-coordinates of its left and right endpoints are $\min\{r(s')\,|\, s'\hbox{ intersects }s\}\cup\{r(s)\}$  and  
$\max\{r(s')\,|\, s'\hbox{ intersects }s\}\cup\{r(s)\}$ respectively.
For each vertical segment $s$ of the initial representation, its new $x$-coordinate is $r(s)$ and the new $y$-coordinates of its bottom and top endpoints are $\min\{r(s')\,|\, s'\hbox{ intersects }s\}\cup \{r(s)\}$ and 
$\max\{r(s')\,|\, s'\hbox{ intersects }s\}\cup \{r(s)\}$ respectively. Then, the segment $s$ intersects the stab line at $(r(s),r(s))$ and any other segment $s'$ it must intersect at $(r(s),r(s'))$ when $s$ is vertical,  respectively at $(r(s'),r(s))$ when $s$ is horizontal. 

Therefore, it is sufficient to prove that the problem of deciding whether an input graph has a good representation is in $\NP$. Due to the integer coordinates, the task of checking whether a given set of segments, one for each vertex, is a good representation of the input graph is polynomial. 
Furthermore, $\apexgraph{G}$  is obtained in polynomial time from $G$. It remains to show that $G$ is a yes-instance of {\sc PHPC} if and only if $\apexgraph{G}$ is a yes-instance of {\sc StabGIGRec}. Proposition \ref{prop:forward} shows the forward direction. 
For the backward direction, we remark -- seeking to use Proposition~\ref{prop:yes-instance} -- that if $\apexgraph{G}$ is a $\StabGIG$, then $\apexgraph{G}$ is a 1-string graph. To this end, it is sufficient to note that each segment in a grid representation of $\apexgraph{G}$ is a 1-string, and that two segments that intersect without crossing (because an endpoint of one segment lies on the second segment) may be easily transformed into two crossing segments. Now, using Proposition~\ref{prop:yes-instance}, we deduce that $G$ is a yes-instance of {\sc PHPC}. This completes the proof of Theorem~\ref{thm:stabbable-GIG}.
\medskip

Observe that the graph $\apexgraph{G}$ used in our proof is 'almost' planar, in the sense that $\apexgraph{G}\setminus\{a\}$ is planar. A graph that can be made planar by removing one of its vertices is called an \emph{apex} graph \cite{welsh1993,thilikos1997}. Moreover, $\apexgraph{G}$ has girth $g=k+3\geq 10$ so that we may deduce the following result.

\begin{corollary}
For each integer $g\geq 10$, {\sc StabGIGRec} remains $\NP$-complete even if the inputs are restricted to apex graphs of girth exactly $g$. 
\label{cor:generalization}
\end{corollary}

\section{Conclusion}\label{sec:conclude}

In this paper, we proved that recognizing grounded \LS-graphs is $\NP$-complete. A natural direction of research would be to find interesting subclasses of grounded \LS-graphs with polynomial time recognition algorithms. One such candidate class could be the \emph{grounded unit \LS-graphs}. An unit \LS-shape is a made by joining the bottom end-point of a vertical  ($\vert$) segment (of arbitrary length) to the left end-point of a horizontal ($-$) segment of \emph{unit} length. Grounded unit \LS-graphs are the intersection graphs of unit \LS-shapes such that all the \LS-shapes' anchors lie on the same horizontal line. Grounded unit \LS-graphs is a subclass of \emph{cocomparability} graphs and contains \emph{unit interval} graphs and \emph{permutation graphs} as its subclasses. The above discussion motivates the following question.

\begin{question}
What is the computational complexity of recognizing grounded unit \LS-graphs?
\end{question}


Another open question concerns grounded square \LS-graphs i.e. intersection graphs of grounded \LS-shapes whose horizontal and vertical segments have the same length~\cite{bose2022computing}. 
%
\begin{question}
What is the computational complexity of recognizing grounded square \LS-graphs?
\end{question}

As noticed in Corollary \ref{cor:generalization}, recognizing $\StabGIG$ remains $\NP$-complete even if the inputs are restricted to apex graphs. Apex graphs do not contain $K_6$ as a minor. In contrast, recognizing $K_4$-minor free stabbable grid intersection graphs is trivial (since $K_4$-minor free graphs are planar and all planar bipartite graphs are stabbable grid intersection graphs). This motivates the following question.

\begin{question}
What is the computational complexity of recognizing $K_5$-minor free stabbable grid intersection graphs?
\end{question}

%

\bibliographystyle{plain}

\input{mainTCS.bbl}
\end{document}

%% file: FigGrounded_DExample.pdf_t
\begin{picture}(0,0)%
\includegraphics{FigGrounded_DExample.pdf}%
\end{picture}%
\setlength{\unitlength}{4144sp}%
\begingroup\makeatletter\ifx\SetFigFont\undefined%
\gdef\SetFigFont#1#2#3#4#5{%
  \reset@font\fontsize{#1}{#2pt}%
  \fontfamily{#3}\fontseries{#4}\fontshape{#5}%
  \selectfont}%
\fi\endgroup%
\begin{picture}(6507,4212)(1744,-6913)
\put(8236,-3301){\makebox(0,0)[lb]{\smash{{\SetFigFont{17}{20.4}{\familydefault}{\mddefault}{\updefault}{\color[rgb]{0,0,0}$\ell$}%
}}}}
\end{picture}%

%% file: StickToL.pdf_t
\begin{picture}(0,0)%
\includegraphics{StickToL.pdf}%
\end{picture}%
\setlength{\unitlength}{4144sp}%
\begingroup\makeatletter\ifx\SetFigFont\undefined%
\gdef\SetFigFont#1#2#3#4#5{%
  \reset@font\fontsize{#1}{#2pt}%
  \fontfamily{#3}\fontseries{#4}\fontshape{#5}%
  \selectfont}%
\fi\endgroup%
\begin{picture}(9702,1824)(661,-2548)
\put(1216,-1501){\makebox(0,0)[lb]{\smash{{\SetFigFont{10}{12.0}{\familydefault}{\mddefault}{\updefault}{\color[rgb]{0,0,0}$a_1$}%
}}}}
\put(1441,-1726){\makebox(0,0)[lb]{\smash{{\SetFigFont{10}{12.0}{\familydefault}{\mddefault}{\updefault}{\color[rgb]{0,0,0}$b_2$}%
}}}}
\put(1711,-1996){\makebox(0,0)[lb]{\smash{{\SetFigFont{10}{12.0}{\familydefault}{\mddefault}{\updefault}{\color[rgb]{0,0,0}$a_2$}%
}}}}
\put(2116,-2401){\makebox(0,0)[lb]{\smash{{\SetFigFont{10}{12.0}{\familydefault}{\mddefault}{\updefault}{\color[rgb]{0,0,0}$a_3$}%
}}}}
\put(3151,-1006){\makebox(0,0)[lb]{\smash{{\SetFigFont{10}{12.0}{\familydefault}{\mddefault}{\updefault}{\color[rgb]{0,0,0}$b_1$}%
}}}}
\put(5851,-1006){\makebox(0,0)[lb]{\smash{{\SetFigFont{10}{12.0}{\familydefault}{\mddefault}{\updefault}{\color[rgb]{0,0,0}$b_1$}%
}}}}
\put(676,-1006){\makebox(0,0)[lb]{\smash{{\SetFigFont{10}{12.0}{\familydefault}{\mddefault}{\updefault}{\color[rgb]{0,0,0}$b_1$}%
}}}}
\put(3826,-1681){\makebox(0,0)[lb]{\smash{{\SetFigFont{10}{12.0}{\familydefault}{\mddefault}{\updefault}{\color[rgb]{0,0,0}$b_2$}%
}}}}
\put(5851,-1681){\makebox(0,0)[lb]{\smash{{\SetFigFont{10}{12.0}{\familydefault}{\mddefault}{\updefault}{\color[rgb]{0,0,0}$b_2$}%
}}}}
\put(5851,-2356){\makebox(0,0)[lb]{\smash{{\SetFigFont{10}{12.0}{\familydefault}{\mddefault}{\updefault}{\color[rgb]{0,0,0}$a_3$}%
}}}}
\put(9406,-2446){\makebox(0,0)[lb]{\smash{{\SetFigFont{10}{12.0}{\familydefault}{\mddefault}{\updefault}{\color[rgb]{0,0,0}$b_2$}%
}}}}
\put(9181,-2446){\makebox(0,0)[lb]{\smash{{\SetFigFont{10}{12.0}{\familydefault}{\mddefault}{\updefault}{\color[rgb]{0,0,0}$a_1$}%
}}}}
\put(8731,-2446){\makebox(0,0)[lb]{\smash{{\SetFigFont{10}{12.0}{\familydefault}{\mddefault}{\updefault}{\color[rgb]{0,0,0}$b_1$}%
}}}}
\put(3556,-1411){\makebox(0,0)[lb]{\smash{{\SetFigFont{10}{12.0}{\familydefault}{\mddefault}{\updefault}{\color[rgb]{0,0,0}$a_1$}%
}}}}
\put(4006,-1861){\makebox(0,0)[lb]{\smash{{\SetFigFont{10}{12.0}{\familydefault}{\mddefault}{\updefault}{\color[rgb]{0,0,0}$a_2$}%
}}}}
\put(4456,-2311){\makebox(0,0)[lb]{\smash{{\SetFigFont{10}{12.0}{\familydefault}{\mddefault}{\updefault}{\color[rgb]{0,0,0}$a_3$}%
}}}}
\put(5851,-1411){\makebox(0,0)[lb]{\smash{{\SetFigFont{10}{12.0}{\familydefault}{\mddefault}{\updefault}{\color[rgb]{0,0,0}$a_1$}%
}}}}
\put(5851,-1861){\makebox(0,0)[lb]{\smash{{\SetFigFont{10}{12.0}{\familydefault}{\mddefault}{\updefault}{\color[rgb]{0,0,0}$a_2$}%
}}}}
\put(9631,-2446){\makebox(0,0)[lb]{\smash{{\SetFigFont{10}{12.0}{\familydefault}{\mddefault}{\updefault}{\color[rgb]{0,0,0}$a_2$}%
}}}}
\put(10081,-2446){\makebox(0,0)[lb]{\smash{{\SetFigFont{10}{12.0}{\familydefault}{\mddefault}{\updefault}{\color[rgb]{0,0,0}$a_3$}%
}}}}
\end{picture}%

%% file: NotStick.pdf_t
\begin{picture}(0,0)%
\includegraphics{NotStick.pdf}%
\end{picture}%
\setlength{\unitlength}{4144sp}%
\begingroup\makeatletter\ifx\SetFigFont\undefined%
\gdef\SetFigFont#1#2#3#4#5{%
  \reset@font\fontsize{#1}{#2pt}%
  \fontfamily{#3}\fontseries{#4}\fontshape{#5}%
  \selectfont}%
\fi\endgroup%
\begin{picture}(17083,5982)(35,-4591)
\put(1219,-2975){\rotatebox{45.0}{\makebox(0,0)[lb]{\smash{{\SetFigFont{14}{16.8}{\familydefault}{\mddefault}{\updefault}{\color[rgb]{0,0,0}$a_1$}%
}}}}}
\put(1760,-2434){\rotatebox{45.0}{\makebox(0,0)[lb]{\smash{{\SetFigFont{14}{16.8}{\familydefault}{\mddefault}{\updefault}{\color[rgb]{0,0,0}$a_2$}%
}}}}}
\put(3828,-365){\rotatebox{45.0}{\makebox(0,0)[lb]{\smash{{\SetFigFont{14}{16.8}{\familydefault}{\mddefault}{\updefault}{\color[rgb]{0,0,0}$a_3$}%
}}}}}
\put(4338,144){\rotatebox{45.0}{\makebox(0,0)[lb]{\smash{{\SetFigFont{14}{16.8}{\familydefault}{\mddefault}{\updefault}{\color[rgb]{0,0,0}$a_4$}%
}}}}}
\put(456,-3738){\rotatebox{45.0}{\makebox(0,0)[lb]{\smash{{\SetFigFont{14}{16.8}{\familydefault}{\mddefault}{\updefault}{\color[rgb]{0,0,0}$b_1$}%
}}}}}
\put(2524,-1670){\rotatebox{45.0}{\makebox(0,0)[lb]{\smash{{\SetFigFont{14}{16.8}{\familydefault}{\mddefault}{\updefault}{\color[rgb]{0,0,0}$b_2$}%
}}}}}
\put(3033,-1161){\rotatebox{45.0}{\makebox(0,0)[lb]{\smash{{\SetFigFont{14}{16.8}{\familydefault}{\mddefault}{\updefault}{\color[rgb]{0,0,0}$b_3$}%
}}}}}
\put(5133,939){\rotatebox{45.0}{\makebox(0,0)[lb]{\smash{{\SetFigFont{14}{16.8}{\familydefault}{\mddefault}{\updefault}{\color[rgb]{0,0,0}$b_4$}%
}}}}}
\put(169,-4279){\rotatebox{45.0}{\makebox(0,0)[lb]{\smash{{\SetFigFont{14}{16.8}{\familydefault}{\mddefault}{\updefault}{\color[rgb]{0,0,0}$x_1$}%
}}}}}
\put(1824,-2943){\rotatebox{45.0}{\makebox(0,0)[lb]{\smash{{\SetFigFont{14}{16.8}{\familydefault}{\mddefault}{\updefault}{\color[rgb]{0,0,0}$y_1$}%
}}}}}
\put(4528,-238){\rotatebox{45.0}{\makebox(0,0)[lb]{\smash{{\SetFigFont{14}{16.8}{\familydefault}{\mddefault}{\updefault}{\color[rgb]{0,0,0}$y_2$}%
}}}}}
\put(3733,-2115){\rotatebox{45.0}{\makebox(0,0)[lb]{\smash{{\SetFigFont{14}{16.8}{\familydefault}{\mddefault}{\updefault}{\color[rgb]{0,0,0}$y_3$}%
}}}}}
\put(3192,-1574){\rotatebox{45.0}{\makebox(0,0)[lb]{\smash{{\SetFigFont{14}{16.8}{\familydefault}{\mddefault}{\updefault}{\color[rgb]{0,0,0}$x_2$}%
}}}}}
\put(5706,1258){\rotatebox{45.0}{\makebox(0,0)[lb]{\smash{{\SetFigFont{14}{16.8}{\familydefault}{\mddefault}{\updefault}{\color[rgb]{0,0,0}$x_3$}%
}}}}}
\put(8636,-4497){\makebox(0,0)[lb]{\smash{{\SetFigFont{14}{16.8}{\familydefault}{\mddefault}{\updefault}{\color[rgb]{0,0,0}$b_1$}%
}}}}
\put(9695,-4497){\makebox(0,0)[lb]{\smash{{\SetFigFont{14}{16.8}{\familydefault}{\mddefault}{\updefault}{\color[rgb]{0,0,0}$a_1$}%
}}}}
\put(10753,-4497){\makebox(0,0)[lb]{\smash{{\SetFigFont{14}{16.8}{\familydefault}{\mddefault}{\updefault}{\color[rgb]{0,0,0}$a_2$}%
}}}}
\put(11812,-4497){\makebox(0,0)[lb]{\smash{{\SetFigFont{14}{16.8}{\familydefault}{\mddefault}{\updefault}{\color[rgb]{0,0,0}$b_2$}%
}}}}
\put(12871,-4497){\makebox(0,0)[lb]{\smash{{\SetFigFont{14}{16.8}{\familydefault}{\mddefault}{\updefault}{\color[rgb]{0,0,0}$b_3$}%
}}}}
\put(13930,-4497){\makebox(0,0)[lb]{\smash{{\SetFigFont{14}{16.8}{\familydefault}{\mddefault}{\updefault}{\color[rgb]{0,0,0}$a_3$}%
}}}}
\put(14989,-4497){\makebox(0,0)[lb]{\smash{{\SetFigFont{14}{16.8}{\familydefault}{\mddefault}{\updefault}{\color[rgb]{0,0,0}$a_4$}%
}}}}
\put(8000,-4497){\makebox(0,0)[lb]{\smash{{\SetFigFont{14}{16.8}{\familydefault}{\mddefault}{\updefault}{\color[rgb]{0,0,0}$x_1$}%
}}}}
\put(10224,-4497){\makebox(0,0)[lb]{\smash{{\SetFigFont{14}{16.8}{\familydefault}{\mddefault}{\updefault}{\color[rgb]{0,0,0}$y_1$}%
}}}}
\put(12342,-4497){\makebox(0,0)[lb]{\smash{{\SetFigFont{14}{16.8}{\familydefault}{\mddefault}{\updefault}{\color[rgb]{0,0,0}$x_2$}%
}}}}
\put(14459,-4497){\makebox(0,0)[lb]{\smash{{\SetFigFont{14}{16.8}{\familydefault}{\mddefault}{\updefault}{\color[rgb]{0,0,0}$y_2$}%
}}}}
\put(15995,-4497){\makebox(0,0)[lb]{\smash{{\SetFigFont{14}{16.8}{\familydefault}{\mddefault}{\updefault}{\color[rgb]{0,0,0}$b_4$}%
}}}}
\put(16312,-4497){\makebox(0,0)[lb]{\smash{{\SetFigFont{14}{16.8}{\familydefault}{\mddefault}{\updefault}{\color[rgb]{0,0,0}$y_3$}%
}}}}
\put(16630,-4497){\makebox(0,0)[lb]{\smash{{\SetFigFont{14}{16.8}{\familydefault}{\mddefault}{\updefault}{\color[rgb]{0,0,0}$x_3$}%
}}}}
\end{picture}%

%% file: GammaAndReduc.pdf_t
\begin{picture}(0,0)%
\includegraphics{GammaAndReduc.pdf}%
\end{picture}%
\setlength{\unitlength}{4144sp}%
\begingroup\makeatletter\ifx\SetFigFont\undefined%
\gdef\SetFigFont#1#2#3#4#5{%
  \reset@font\fontsize{#1}{#2pt}%
  \fontfamily{#3}\fontseries{#4}\fontshape{#5}%
  \selectfont}%
\fi\endgroup%
\begin{picture}(14485,3945)(1357,-2191)
\put(11026,524){\makebox(0,0)[lb]{\smash{{\SetFigFont{14}{16.8}{\rmdefault}{\mddefault}{\updefault}{\color[rgb]{0,0,0}$\Lambda(b_1)$}%
}}}}
\put(14086,524){\makebox(0,0)[lb]{\smash{{\SetFigFont{14}{16.8}{\rmdefault}{\mddefault}{\updefault}{\color[rgb]{0,0,0}$\Lambda(b_2)$}%
}}}}
\put(9406,-1051){\makebox(0,0)[lb]{\smash{{\SetFigFont{14}{16.8}{\rmdefault}{\mddefault}{\updefault}{\color[rgb]{0,0,0}$a_1$}%
}}}}
\put(12646,-1051){\makebox(0,0)[lb]{\smash{{\SetFigFont{14}{16.8}{\rmdefault}{\mddefault}{\updefault}{\color[rgb]{0,0,0}$a_2$}%
}}}}
\put(15706,-1051){\makebox(0,0)[lb]{\smash{{\SetFigFont{14}{16.8}{\rmdefault}{\mddefault}{\updefault}{\color[rgb]{0,0,0}$a_3$}%
}}}}
\end{picture}%

%% file: Gamma2Rep.pdf_t
\begin{picture}(0,0)%
\includegraphics{Gamma2Rep.pdf}%
\end{picture}%
\setlength{\unitlength}{4144sp}%
\begingroup\makeatletter\ifx\SetFigFont\undefined%
\gdef\SetFigFont#1#2#3#4#5{%
  \reset@font\fontsize{#1}{#2pt}%
  \fontfamily{#3}\fontseries{#4}\fontshape{#5}%
  \selectfont}%
\fi\endgroup%
\begin{picture}(23636,5912)(6122,-4288)
\put(12916,-2851){\makebox(0,0)[lb]{\smash{{\SetFigFont{34}{40.8}{\rmdefault}{\mddefault}{\updefault}{\color[rgb]{0,0,0}$v_f^{'\, 1}$}%
}}}}
\put(9541,-2851){\makebox(0,0)[lb]{\smash{{\SetFigFont{34}{40.8}{\rmdefault}{\mddefault}{\updefault}{\color[rgb]{0,0,0}$e$}%
}}}}
\put(8866,-2851){\makebox(0,0)[lb]{\smash{{\SetFigFont{34}{40.8}{\rmdefault}{\mddefault}{\updefault}{\color[rgb]{0,0,0}$v_d^{2}$}%
}}}}
\put(11611,-2851){\makebox(0,0)[lb]{\smash{{\SetFigFont{34}{40.8}{\rmdefault}{\mddefault}{\updefault}{\color[rgb]{0,0,0}$v_d^{1}$}%
}}}}
\put(12241,-2851){\makebox(0,0)[lb]{\smash{{\SetFigFont{34}{40.8}{\rmdefault}{\mddefault}{\updefault}{\color[rgb]{0,0,0}$v_e^{1}$}%
}}}}
\put(8281,-2851){\makebox(0,0)[lb]{\smash{{\SetFigFont{34}{40.8}{\rmdefault}{\mddefault}{\updefault}{\color[rgb]{0,0,0}$d$}%
}}}}
\put(7111,-2851){\makebox(0,0)[lb]{\smash{{\SetFigFont{34}{40.8}{\rmdefault}{\mddefault}{\updefault}{\color[rgb]{0,0,0}$v_ e^{2}$}%
}}}}
\put(7741,-2851){\makebox(0,0)[lb]{\smash{{\SetFigFont{34}{40.8}{\rmdefault}{\mddefault}{\updefault}{\color[rgb]{0,0,0}$v_f^{2}$}%
}}}}
\put(10936,-2851){\makebox(0,0)[lb]{\smash{{\SetFigFont{34}{40.8}{\rmdefault}{\mddefault}{\updefault}{\color[rgb]{0,0,0}$f$}%
}}}}
\put(10261,-2851){\makebox(0,0)[lb]{\smash{{\SetFigFont{34}{40.8}{\rmdefault}{\mddefault}{\updefault}{\color[rgb]{0,0,0}$v_f^{1}$}%
}}}}
\put(6481,-2851){\makebox(0,0)[lb]{\smash{{\SetFigFont{34}{40.8}{\rmdefault}{\mddefault}{\updefault}{\color[rgb]{0,0,0}$v_d^{'\,2}$}%
}}}}
\put(9406,-4156){\makebox(0,0)[lb]{\smash{{\SetFigFont{25}{30.0}{\rmdefault}{\mddefault}{\updefault}{\color[rgb]{0,0,0}(a)}%
}}}}
\put(19441,-2851){\makebox(0,0)[lb]{\smash{{\SetFigFont{34}{40.8}{\rmdefault}{\mddefault}{\updefault}{\color[rgb]{0,0,0}$v_d^{1}$}%
}}}}
\put(28756,-2851){\makebox(0,0)[lb]{\smash{{\SetFigFont{34}{40.8}{\rmdefault}{\mddefault}{\updefault}{\color[rgb]{0,0,0}$x$}%
}}}}
\put(28576,-4111){\makebox(0,0)[lb]{\smash{{\SetFigFont{25}{30.0}{\rmdefault}{\mddefault}{\updefault}{\color[rgb]{0,0,0}(d)}%
}}}}
\put(18091,-2851){\makebox(0,0)[lb]{\smash{{\SetFigFont{34}{40.8}{\rmdefault}{\mddefault}{\updefault}{\color[rgb]{0,0,0}$v_f^{1}$}%
}}}}
\put(18721,-2851){\makebox(0,0)[lb]{\smash{{\SetFigFont{34}{40.8}{\rmdefault}{\mddefault}{\updefault}{\color[rgb]{0,0,0}$f$}%
}}}}
\put(14311,-2851){\makebox(0,0)[lb]{\smash{{\SetFigFont{34}{40.8}{\rmdefault}{\mddefault}{\updefault}{\color[rgb]{0,0,0}$v_e^1$}%
}}}}
\put(14896,-2851){\makebox(0,0)[lb]{\smash{{\SetFigFont{34}{40.8}{\rmdefault}{\mddefault}{\updefault}{\color[rgb]{0,0,0}$v_ e^{2}$}%
}}}}
\put(15571,-2851){\makebox(0,0)[lb]{\smash{{\SetFigFont{34}{40.8}{\rmdefault}{\mddefault}{\updefault}{\color[rgb]{0,0,0}$v_f^{2}$}%
}}}}
\put(16741,-2851){\makebox(0,0)[lb]{\smash{{\SetFigFont{34}{40.8}{\rmdefault}{\mddefault}{\updefault}{\color[rgb]{0,0,0}$v_d^{2}$}%
}}}}
\put(16156,-2851){\makebox(0,0)[lb]{\smash{{\SetFigFont{34}{40.8}{\rmdefault}{\mddefault}{\updefault}{\color[rgb]{0,0,0}$d$}%
}}}}
\put(17371,-2851){\makebox(0,0)[lb]{\smash{{\SetFigFont{34}{40.8}{\rmdefault}{\mddefault}{\updefault}{\color[rgb]{0,0,0}$e$}%
}}}}
\put(16696,-4156){\makebox(0,0)[lb]{\smash{{\SetFigFont{25}{30.0}{\rmdefault}{\mddefault}{\updefault}{\color[rgb]{0,0,0}(b)}%
}}}}
\put(20836,-2851){\makebox(0,0)[lb]{\smash{{\SetFigFont{34}{40.8}{\rmdefault}{\mddefault}{\updefault}{\color[rgb]{0,0,0}$10$}%
}}}}
\put(22681,-2851){\makebox(0,0)[lb]{\smash{{\SetFigFont{34}{40.8}{\rmdefault}{\mddefault}{\updefault}{\color[rgb]{0,0,0}$d$}%
}}}}
\put(21421,-2851){\makebox(0,0)[lb]{\smash{{\SetFigFont{34}{40.8}{\rmdefault}{\mddefault}{\updefault}{\color[rgb]{0,0,0}$v_ e^{2}$}%
}}}}
\put(22096,-2851){\makebox(0,0)[lb]{\smash{{\SetFigFont{34}{40.8}{\rmdefault}{\mddefault}{\updefault}{\color[rgb]{0,0,0}$v_f^{2}$}%
}}}}
\put(23221,-2851){\makebox(0,0)[lb]{\smash{{\SetFigFont{34}{40.8}{\rmdefault}{\mddefault}{\updefault}{\color[rgb]{0,0,0}$v_d^2$}%
}}}}
\put(23896,-2851){\makebox(0,0)[lb]{\smash{{\SetFigFont{34}{40.8}{\rmdefault}{\mddefault}{\updefault}{\color[rgb]{0,0,0}$e$}%
}}}}
\put(24616,-2851){\makebox(0,0)[lb]{\smash{{\SetFigFont{34}{40.8}{\rmdefault}{\mddefault}{\updefault}{\color[rgb]{0,0,0}$v_f^{1}$}%
}}}}
\put(25246,-2851){\makebox(0,0)[lb]{\smash{{\SetFigFont{34}{40.8}{\rmdefault}{\mddefault}{\updefault}{\color[rgb]{0,0,0}$f$}%
}}}}
\put(23176,-3571){\makebox(0,0)[lb]{\smash{{\SetFigFont{34}{40.8}{\rmdefault}{\mddefault}{\updefault}{\color[rgb]{0,0,0}($e=2=x$)}%
}}}}
\put(23761,-4156){\makebox(0,0)[lb]{\smash{{\SetFigFont{25}{30.0}{\rmdefault}{\mddefault}{\updefault}{\color[rgb]{0,0,0}(c)}%
}}}}
\put(25831,-2806){\makebox(0,0)[lb]{\smash{{\SetFigFont{34}{40.8}{\rmdefault}{\mddefault}{\updefault}{\color[rgb]{0,0,0}$v_d^{1}$}%
}}}}
\put(26506,-2806){\makebox(0,0)[lb]{\smash{{\SetFigFont{34}{40.8}{\rmdefault}{\mddefault}{\updefault}{\color[rgb]{0,0,0}$v_e^{1}$}%
}}}}
\end{picture}%

%% file: GroundedH.pdf_t
\begin{picture}(0,0)%
\includegraphics{GroundedH.pdf}%
\end{picture}%
\setlength{\unitlength}{4144sp}%
\begingroup\makeatletter\ifx\SetFigFont\undefined%
\gdef\SetFigFont#1#2#3#4#5{%
  \reset@font\fontsize{#1}{#2pt}%
  \fontfamily{#3}\fontseries{#4}\fontshape{#5}%
  \selectfont}%
\fi\endgroup%
\begin{picture}(5671,1938)(-33,-1606)
\put(4906,-1546){\makebox(0,0)[lb]{\smash{{\SetFigFont{10}{12.0}{\familydefault}{\mddefault}{\updefault}{\color[rgb]{0,0,0}$a_2$}%
}}}}
\put(5356,-1546){\makebox(0,0)[lb]{\smash{{\SetFigFont{10}{12.0}{\familydefault}{\mddefault}{\updefault}{\color[rgb]{0,0,0}$a_3$}%
}}}}
\put(3556,-1546){\makebox(0,0)[lb]{\smash{{\SetFigFont{10}{12.0}{\familydefault}{\mddefault}{\updefault}{\color[rgb]{0,0,0}$b_1$}%
}}}}
\put(4006,-1546){\makebox(0,0)[lb]{\smash{{\SetFigFont{10}{12.0}{\familydefault}{\mddefault}{\updefault}{\color[rgb]{0,0,0}$a_1$}%
}}}}
\put(4501,-1546){\makebox(0,0)[lb]{\smash{{\SetFigFont{10}{12.0}{\familydefault}{\mddefault}{\updefault}{\color[rgb]{0,0,0}$b_2$}%
}}}}
\put(834,-1546){\makebox(0,0)[lb]{\smash{{\SetFigFont{10}{12.0}{\familydefault}{\mddefault}{\updefault}{\color[rgb]{0,0,0}$b_2$}%
}}}}
\put(609,-1546){\makebox(0,0)[lb]{\smash{{\SetFigFont{10}{12.0}{\familydefault}{\mddefault}{\updefault}{\color[rgb]{0,0,0}$a_1$}%
}}}}
\put(159,-1546){\makebox(0,0)[lb]{\smash{{\SetFigFont{10}{12.0}{\familydefault}{\mddefault}{\updefault}{\color[rgb]{0,0,0}$b_1$}%
}}}}
\put(1059,-1546){\makebox(0,0)[lb]{\smash{{\SetFigFont{10}{12.0}{\familydefault}{\mddefault}{\updefault}{\color[rgb]{0,0,0}$a_2$}%
}}}}
\put(1509,-1546){\makebox(0,0)[lb]{\smash{{\SetFigFont{10}{12.0}{\familydefault}{\mddefault}{\updefault}{\color[rgb]{0,0,0}$a_3$}%
}}}}
\end{picture}%

%% file: FigEdgeRep.pdf_t
\begin{picture}(0,0)%
\includegraphics{FigEdgeRep.pdf}%
\end{picture}%
\setlength{\unitlength}{4144sp}%
\begingroup\makeatletter\ifx\SetFigFont\undefined%
\gdef\SetFigFont#1#2#3#4#5{%
  \reset@font\fontsize{#1}{#2pt}%
  \fontfamily{#3}\fontseries{#4}\fontshape{#5}%
  \selectfont}%
\fi\endgroup%
\begin{picture}(15798,12159)(-1319,-10828)
\put(4231,-3436){\makebox(0,0)[lb]{\smash{{\SetFigFont{20}{24.0}{\rmdefault}{\mddefault}{\updefault}{\color[rgb]{0,0,0}$\epsilon d_I$}%
}}}}
\put(4771,-5236){\makebox(0,0)[lb]{\smash{{\SetFigFont{20}{24.0}{\rmdefault}{\mddefault}{\updefault}{\color[rgb]{0,0,0}$-i+0.01$}%
}}}}
\put(4231,-6001){\makebox(0,0)[lb]{\smash{{\SetFigFont{20}{24.0}{\rmdefault}{\mddefault}{\updefault}{\color[rgb]{0,0,0}$\mathtt{c}(u_e^1)$}%
}}}}
\put(14176,-5281){\makebox(0,0)[lb]{\smash{{\SetFigFont{20}{24.0}{\rmdefault}{\mddefault}{\updefault}{\color[rgb]{0,0,0}$x$-axis}%
}}}}
\put(6751,-2626){\makebox(0,0)[lb]{\smash{{\SetFigFont{20}{24.0}{\rmdefault}{\mddefault}{\updefault}{\color[rgb]{0,0,0}$w_i(0,i-1)$}%
}}}}
\put(6751,-1276){\makebox(0,0)[lb]{\smash{{\SetFigFont{20}{24.0}{\rmdefault}{\mddefault}{\updefault}{\color[rgb]{0,0,0}$w_j(0,j-1)$}%
}}}}
\put(10981,-1186){\makebox(0,0)[lb]{\smash{{\SetFigFont{20}{24.0}{\rmdefault}{\mddefault}{\updefault}{\color[rgb]{0,0,0}$(2n_H-j+0.1,j-1)$}%
}}}}
\put(6796,1064){\makebox(0,0)[lb]{\smash{{\SetFigFont{20}{24.0}{\rmdefault}{\mddefault}{\updefault}{\color[rgb]{0,0,0}$y$-axis}%
}}}}
\put(1261,-2086){\makebox(0,0)[lb]{\smash{{\SetFigFont{20}{24.0}{\rmdefault}{\mddefault}{\updefault}{\color[rgb]{0,0,0}$\epsilon d_I$}%
}}}}
\put(2701,-9736){\makebox(0,0)[lb]{\smash{{\SetFigFont{20}{24.0}{\rmdefault}{\mddefault}{\updefault}{\color[rgb]{0,0,0}$\mathtt{c}(u_e^2)$}%
}}}}
\put(271,-5236){\makebox(0,0)[lb]{\smash{{\SetFigFont{20}{24.0}{\rmdefault}{\mddefault}{\updefault}{\color[rgb]{0,0,0}$-j-0.01$}%
}}}}
\put(-1304,-1231){\makebox(0,0)[lb]{\smash{{\SetFigFont{20}{24.0}{\rmdefault}{\mddefault}{\updefault}{\color[rgb]{0,0,0}$(-j-0.1,j-1)$}%
}}}}
\put(1126,-1186){\makebox(0,0)[lb]{\smash{{\SetFigFont{20}{24.0}{\rmdefault}{\mddefault}{\updefault}{\color[rgb]{0,0,0}$(-j,j-1)$}%
}}}}
\put(1891,-2626){\makebox(0,0)[lb]{\smash{{\SetFigFont{20}{24.0}{\rmdefault}{\mddefault}{\updefault}{\color[rgb]{0,0,0}$(-i-0.1,i-1)$}%
}}}}
\put(4276,-2626){\makebox(0,0)[lb]{\smash{{\SetFigFont{20}{24.0}{\rmdefault}{\mddefault}{\updefault}{\color[rgb]{0,0,0}$(-i,i-1)$}%
}}}}
\put(3736,-1231){\makebox(0,0)[lb]{\smash{{\SetFigFont{20}{24.0}{\rmdefault}{\mddefault}{\updefault}{\color[rgb]{0,0,0}$\mathtt{c}(w_j)$}%
}}}}
\put(13186,-2626){\makebox(0,0)[lb]{\smash{{\SetFigFont{20}{24.0}{\rmdefault}{\mddefault}{\updefault}{\color[rgb]{0,0,0}$(2n_H-i+0.1,i-1)$}%
}}}}
\put(9991,-2581){\makebox(0,0)[lb]{\smash{{\SetFigFont{20}{24.0}{\rmdefault}{\mddefault}{\updefault}{\color[rgb]{0,0,0}$\mathtt{c}(w_i)$}%
}}}}
\put(1576,-10276){\makebox(0,0)[lb]{\smash{{\SetFigFont{20}{24.0}{\rmdefault}{\mddefault}{\updefault}{\color[rgb]{0,0,0}$(\mathbf{j}_e,\mathbf{j}_e)$}%
}}}}
\put(4006,-10321){\makebox(0,0)[lb]{\smash{{\SetFigFont{20}{24.0}{\rmdefault}{\mddefault}{\updefault}{\color[rgb]{0,0,0}$(\mathbf{i}_e,\mathbf{j}_e)$}%
}}}}
\put(811,-7171){\makebox(0,0)[lb]{\smash{{\SetFigFont{20}{24.0}{\rmdefault}{\mddefault}{\updefault}{\color[rgb]{0,0,0}$\mathtt{c}(u_e^3)$}%
}}}}
\put(5761,164){\makebox(0,0)[lb]{\smash{{\SetFigFont{20}{24.0}{\rmdefault}{\mddefault}{\updefault}{\color[rgb]{0,0,0}$\mathtt{c}(a)$}%
}}}}
\put(12376,974){\makebox(0,0)[lb]{\smash{{\SetFigFont{20}{24.0}{\rmdefault}{\mddefault}{\updefault}{\color[rgb]{0,0,0}$\ell$}%
}}}}
\put(2026,-4426){\makebox(0,0)[lb]{\smash{{\SetFigFont{20}{24.0}{\rmdefault}{\mddefault}{\updefault}{\color[rgb]{0,0,0}$-j+0.01$}%
}}}}
\put(2791,-5731){\makebox(0,0)[lb]{\smash{{\SetFigFont{20}{24.0}{\rmdefault}{\mddefault}{\updefault}{\color[rgb]{0,0,0}$-i-0.01$}%
}}}}
\end{picture}%